\documentclass[11pt, UKenglish, cleveref, autoref, thm-restate]{article}

\usepackage{amsmath}
\usepackage{amsthm}
\usepackage{amssymb}
\usepackage{algorithm}
\usepackage[noend]{algpseudocode}
\usepackage{color}
\usepackage{xcolor}
\usepackage{babel}

\usepackage{graphicx}
\usepackage{caption}
\usepackage{subcaption}

\usepackage[square,sort,comma,numbers]{natbib}
\usepackage{wrapfig,epsfig}
\usepackage{psfrag}
\usepackage{epstopdf}
\usepackage{url}
\usepackage{graphicx}
\usepackage{epstopdf}
\usepackage{algpseudocode}
\usepackage{thmtools}
\usepackage{thm-restate}
\usepackage{tikz}
\usepackage{hyperref}
\hypersetup{colorlinks=true}
\usepackage[margin=1in]{geometry}

\usepackage{tabularx}

\newcommand{\norm}[1]{\|#1\|}

\newcommand{\Tr}{\mathsf{Tr}}

\newcommand{\poly}{\mathsf{poly}}
\newcommand{\Var}{\mathsf{Var}}

\DeclareMathOperator*{\E}{\mathbb{E}}

\newcommand{\rank}{rank}

\newtheorem{theorem}{Theorem}[section]
\newtheorem{lemma}[theorem]{Lemma}

\newtheorem{definition}[theorem]{Definition}

\newtheorem{claim}[theorem]{Claim}

\newenvironment{proofof}[1]{\bigskip \noindent {\it Proof of #1.}\quad }
{\qed\par\vskip 4mm\par}
\newenvironment{claimproofof}[1]{\bigskip \noindent {\it Proof of #1.}\quad }
{\qed\par\vskip 4mm\par}

\title{Effective Resistances in Non-Expander Graphs}

\author{
    Dongrun Cai\\
    \texttt{cdr@mail.ustc.edu.cn}
    \and
    Xue Chen\\
    \texttt{xuechen1989@ustc.edu.cn}\\
    \and
    Pan Peng\thanks{Supported in part by NSFC grant 62272431 and ``the Fundamental Research Funds for the Central Universities''}\\
    \texttt{ppeng@ustc.edu.cn}
}

\date{School of Computer Science and Technology,\\
University of Science and Technology of China, China}

\begin{document}

\maketitle

\begin{abstract}
    Effective resistances are ubiquitous in graph algorithms and network analysis. For an undirected graph $G$, its effective resistance $R_G(s,t)$ between two vertices $s$ and $t$ is defined as the equivalent resistance between $s$ and $t$ if $G$ is thought of as an electrical network with unit resistance on each edge. If we use $L_G$ to denote the Laplacian matrix of $G$ and $L_G^{\dagger}$ to denote its pseudo-inverse, we have $R_G(s,t)=(\mathbf{1}_s-\mathbf{1}_t)^{\top} L^{\dagger} (\mathbf{1}_s-\mathbf{1}_t)$ such that classical Laplacian solvers \cite{SpielmanT14} provide almost-linear time algorithms to approximate $R_G(s,t)$.
    
    In this work, we study \emph{sublinear} time algorithms to approximate the effective resistance of an \emph{adjacent pair} $s$ and $t$. We consider the classical adjacency list model \cite{ron2019sublinear} for local algorithms. While recent works \cite{andoni2018solving,peng2021local,li2023new} have provided sublinear time algorithms for \emph{expander graphs}, we prove several lower bounds for \emph{general graphs} of $n$ vertices and $m$ edges:
    \begin{enumerate}
        \item It needs $\Omega(n)$ queries to obtain $1.01$-approximations of the effective resistance of an adjacent pair $s$ and $t$, even for graphs of degree at most 3 except $s$ and $t$.
        
        \item For graphs of degree at most $d$ and any parameter $\ell$, it needs $\Omega(m/\ell)$ queries to obtain $c \cdot \min\{d, \ell\}$-approximations where $c>0$ is a universal constant.
    \end{enumerate}
    Moreover, we supplement the first lower bound by providing a sublinear time $(1+\epsilon)$-approximation algorithm for graphs of \emph{degree 2} except the pair $s$ and $t$.

    One of our technical ingredients is to bound the expansion of a graph in terms of the smallest non-trivial eigenvalue of its Laplacian matrix after removing  edges. We discover  a new lower bound on the eigenvalues of perturbed graphs (resp. perturbed matrices) by incorporating the effective resistance of the removed edge (resp. the leverage scores of the removed rows), which may be of independent interest.
\end{abstract}

\section{Introduction}\label{sec:intro}
Effective resistance of a graph is a fundamental quantity to measure the similarity between two vertices. Given an unweighted graph $G$ and two vertices $s$ and $t$, the $s$-$t$ effective resistance, denoted by $R_G(s,t)$, is defined as the electrical distance between $s$ and $t$ when $G$ represents an electrical circuit with each edge $e$ a resistor with electrical resistance $1$. Together with the related concept electrical flows, effective resistances have played important roles in advancing the development of graph algorithms. They have been utilized for computing and approximating maximum flow~\cite{ChristianoKMST11,Madry16}, generating random spanning tree~\cite{MST15fast,Sch17almost}, designing faster algorithms for multicommodity flow~\cite{KMP12faster}, and graph sparsification~\cite{SpielmanS11,DinitzKW15}. In addition, effective resistances have also found applications in machine learning and social network analysis. For example, it has been used for graph convolutional networks~\cite{AHMAD2021389}, for measuring the similarity of vertices in social networks~\cite{LZ18:kirchhoff,peng2021local} and measuring robustness of networks~\cite{ellens2011effective}. 

Let $L_G$ denote the Laplacian matrix of $G$ and $L_G^{\dagger}$ denote its pseudo-inverse. Then the $s$-$t$ effective resistance admits an elegant expression $R_G(s,t)=(\mathbf{1}_s-\mathbf{1}_t)^{\top} \cdot L_G^{\dagger} \cdot (\mathbf{1}_s-\mathbf{1}_t)$ where $\mathbf{1}_u \in \mathbb{R}^n$ denotes the indicator vector at vertex $u$. Hence classical Laplacian solvers~\cite{SpielmanT14,CohenKMPPRX14} provide \emph{almost-linear} time algorithms to approximate $R_G(s,t)$.

It has recently received increasing interest of designing \emph{sublinear-time} (or local) algorithms for estimating effective resistances. In this setting, we are given query access to a graph and any specified vertex pair $s$ and $t$, our goal is to find a good approximation of the $s$-$t$ effective resistance, by making as few queries as possible. Such algorithms are particularly motivated by the ubiquitousness of modern massive graphs, on which traditional polynomial-time algorithms are no longer feasible. 
Andoni et al. \cite{andoni2018solving} gave an algorithm that $(1+\epsilon)$-approximates $R_G(s,t)$ in $O(\frac{1}{\epsilon^2}\poly\log \frac{1}{\epsilon})$ time for $d$-regular expander graphs. 
Peng et al. \cite{peng2021local} then generalized this algorithm to unbounded-degree expander graphs with an additive error $\epsilon$ and similar query complexity. Li and Sachdeva \cite{li2023new} then gave one algorithm that $(1+\epsilon)$-approximates the $s$-$t$ effective resistance in $O(\frac{\poly(\log n)}{\epsilon})$ time on expander graphs (which is implicit in the proof of Theorem 3.1 in \cite{li2023new}). However, the question of whether it is possible to obtain sublinear-time estimation for effective resistances on non-expander or general graphs remains largely open.

Besides the aforementioned sublinear-time algorithms for effective resistances on expander graphs, one can observe an $\Omega(n)$-query lower bound for approximating the effective resistance of \emph{non-adjacent} pair $s$ and $t$. 
Indeed, consider an $n$-vertex path, on which the $s$-$t$ effective resistance is equivalent to the $s$-$t$ shortest path length. Intuitively, for the latter problem, any algorithm with a constant approximation ratio needs to well estimate the number of edges on the path from $s$ to $t$, which takes $\Omega(n)$ queries in the worst case.

In this paper, we consider the power and limitations of sublinear-time algorithms for $s$-$t$ effective resistance such that $s$ and $t$ are \emph{adjacent}, i.e., $(s,t)\in E(G)$. The adjacency case is already interesting in many applications. For example, in the seminal work on graph sparsification \cite{SpielmanS11}, it suffices to have good estimations of the effective resistances between the endpoints of edges, i.e., the adjacent pairs. It is also known that the effective resistance multiplied by the edge weight is equal to the probability that the edge
belongs to a randomly generated spanning tree (see e.g. \cite{DurfeeKPRS16}), which has found applications in random spanning tree generation.

On the other hand, for an adjacency pair, the lower bound from the previous discussion does not hold any more, as for any pair $s$ and $t$ such that $(s,t)\in E$, their effective resistance is exactly $1$ on a path. A priori, it could be true that a $(1+\epsilon)$-approximation of the $s$-$t$ effective resistance for an adjacent pair can be found in sublinear time. 

Now we state our results on sublinear-time algorithms for estimating the effective resistances. We will focus on unit-weighted graphs and the adjacency list model \cite{ron2019sublinear}, in which the local algorithms can perform degree,  neighbor queries and also sample vertices uniformly at random. 





\vspace{.1in}
\noindent\textbf{Strong Lower Bounds.} First, we provide a strong lower bound for graphs of degree at most $3$ except the given pair $s$ and $t$. For convenience, for a parameter $C>1$, we say that value $a$ is a $C$-approximation of value $b$ if $a \in [b/C, b]$. In the following, we always assume that the given pair $s$ and $t$ are \emph{adjacent}, i.e., $(s,t)\in E$. 

\begin{restatable}[]{theorem}{lowerboundsmalldegree} \label{thm:lower_bound_small_degree}
    There are infinitely many $n$ and graphs of $n$ vertices such that any local algorithm with success probability $0.6$ and approximation ratio $1.01$ on $R_G(s,t)$ needs $\Omega(n)$ queries. This holds even for graphs whose vertices are of degree at most $3$ except the adjacent pair $s$ and $t$.
\end{restatable}

Next, observe that the trivial  algorithm outputting 1 directly gives an approximation factor $\frac{2}{1/d(s)+1/d(t)}$ where $d(s)$ and $d(t)$ are the degrees of the adjacent pair $s$ and $t$. This is because $R_G(s,t) \ge \frac{1/d(s)+1/d(t)}{2}$ from the spectral graph theory \cite{lovasz1993random}. 
This factor equals $d$ if the graph is regular and is between $\min\{d(s),d(t)\}$ and $2\min\{d(s),d(t)\}$ in general. If we consider large approximation factors instead of the $(1+\epsilon)$-approximation shown in Theorem~\ref{thm:lower_bound_small_degree}, a natural question is can we design sublinear time algorithms that improve upon this trivial approximation ratio $\frac{2}{1/d(s)+1/d(t)}$? Our next theorem shows that it would take $\Omega(n)$ queries to improve this ratio significantly. 
\begin{theorem}\label{thm:lower_bound_large_degree}[Informal version of Theorem~\ref{thm:formal_lower_bound}]
    There exist a universal constant $c>0$ and infinitely many $n$ such that given any $d\ge 4$ and any $\ell \ge 4$, for graphs of $n$ vertices and degree at most $d$, any local algorithm to approximate $R_G(s,t)$ with success probability $0.6$ and approximation ratio $1 + c \cdot \min\{ d,\ell \}$ needs $\Omega(nd/\ell)$ queries.
\end{theorem}
If we set $\ell=d$, this indicates that it takes $\Omega(n)$ queries to obtain an approximation ratio $o(d)$ for $d$-regular graphs since $m=nd/2$. In contrast, the corresponding trivial approximation ratio for adjacent pairs on $d$-regular graphs is $d$. If we fix $\ell$ and consider graphs with sufficiently large degree $d>\ell$, Theorem~\ref{thm:lower_bound_large_degree} also gives a trade-off between the query complexity and approximation ratio. Next we remark that we could incorporate additive errors into the lower bounds of Theorem~\ref{thm:lower_bound_small_degree} and Theorem~\ref{thm:lower_bound_large_degree}. This is because these two bounds consider the task of distinguishing between graphs with $R_G(s,t)=1$ and graphs with $R_G(s,t)<1$. 

On the other hand, for graphs of degree at most $2$ except the given pair $s$ and $t$, we provide a $(1+\epsilon)$-approximation algorithm of sublinear time in Appendix~\ref{sec:approx_alg}.

\vspace{.1in}
\noindent\textbf{Total effective resistance.}
Another important measure of a network is the total effective resistance, which is defined as $R_{tot}(G)=\sum_{u<v} R_G(u,v)$ and is also known as the Kirchhoff index of a graph \big(see e.g. \cite{ghosh2008minimizing,ellens2011effective,LZ18:kirchhoff} for numerous applications of $R_{tot}(G)$\big). Because of the elegant expression $R_{tot}(G)=n \cdot \sum_{i>1} \frac{1}{\lambda_i(L_G)}$ where the sum is over all non-trivial eigenvalues of the Laplacian of $G$, one may wonder whether there is a simpler approximation algorithm for the total effective resistance or not. However, we show that even for simple graphs with degree $2$, any algorithm with approximation ratio $<2$ needs $\Omega(n)$ queries.

\begin{theorem}\label{thm:total_tight_lower_bound}[Informal version of Theorem~\ref{thm:inapprox_total_eff_res}]
    For any $n$ and any $\ell>1$, any local algorithm for computing the total effective resistance with success probability 0.6 and approximation ratio $\ell$ needs $\Omega(n)$ queries. In particular, this holds even for $\ell=2-o(1)$ and graphs of degree at most $2$.
\end{theorem}

\noindent\textbf{Eigenvalues of graph perturbation.} One technical ingredient in the proof of Theorem~\ref{thm:lower_bound_small_degree} is to show that an expander graph is still an expander after removing one edge. A natural approach would be to compare the smallest non-trivial eigenvalue of the Laplacian of the perturbed graph to the corresponding one of the original graph. Since deleting one edge in $G$ does not change its edge expansion too much, one may use Cheeger's inequality (with other properties) to lower bound the new eigenvalue in terms of the original one. However, our key technical lemma provides a more direct bound by incorporating the effective resistance of the moved edge.

\begin{restatable}[]{lemma}{LapEigenvaluePerturbation} \label{lem:removing_edge_eff_res}
    Given a graph $G=(V,E)$ with $n$ vertices, let $\lambda_1(G) \le \cdots \le \lambda_n(G)$ be the eigenvalues of its Laplacian $L_G$. Given any edge $(u,v)$ in $G$, let $G'=\big( V,E \setminus \{(u,v)\} \big)$ be the graph obtained from $G$ by removing the edge $(u,v)$ and let $\lambda_1(G') \le \ldots \le \lambda_n(G')$ be the eigenvalues of its Laplacian $L_{G'}$. Then it holds that 
    \[
        \forall i \in [n],  \lambda_i(G') \ge \big( 1-R_G(u,v) \big) \cdot \lambda_i(G).
    \]
\end{restatable}

This lemma shows that after removing edge $(u,v)$ in an expander $G$ with $R_G(u,v)$ strictly less than 1, $G$ is still an expander. We remark that this requirement on the effective resistance $R_G(u,v)$ is necessary. Because there exists an expander graph $G$ with edges of effective resistance 1 and vertices of degree 1 (see Claim~\ref{clm:expander_unit_effect_resis} in Section~\ref{sec:proof_eigenvalues_perturbation}), one can not delete an arbitrary edge in $G$ while maintaining the expansion property.

We could consider the general problem of bounding the eigenvalues of a matrix after removing a few rows. Specifically, given a matrix $A \in \mathbb{R}^{m \times n}$ and $k$ distinct rows $\mathbf{a}_{\ell_1},\ldots,\mathbf{a}_{\ell_k}$, the question is to compare all eigenvalues of $(A')^{\top} A'$, where $A' \in \mathbb{R}^{(m-k) \times n}$ removes row $\mathbf{a}_{\ell_1},\ldots,\mathbf{a}_{\ell_k}$ from $A$, to the eigenvalues of $A^{\top} A$. Since $(A')^{\top} A' \preceq	A^{\top} A$, each eigenvalue $\lambda_i\big( (A')^{\top} A' \big) \le \lambda_i\big( A^{\top} A \big)$. We give a lower bound for every $\lambda_i\big( (A')^{\top} A' \big)$ by incorporating the leverage scores of those rows.

The (statistical) leverage scores of a matrix $A \in \mathbb{R}^{m \times n}$ provide a nonuniform importance sampling distribution over the $m$ rows of $A$, which plays a crucial role in randomized matrix algorithms (see \cite{Woodruff14} for a list of applications). For each row $\mathbf{a}_i$ in $A$, its leverage score $\tau_i$ is defined to be $\mathbf{a}_i^{\top} \cdot (A^\top A)^{\dagger} \cdot \mathbf{a}_i$. In fact, the leverage scores of a matrix are the analogues of the effective resistances of a graph~\cite{DrineasM2010}.
More formally, let $B \in \mathbb{R}^{m \times n}$ denote the edge-incidence matrix of a graph $G$ as follows: Each edge $(u,v)$ of $G$ gives a row $\mathbf{1}_u-\mathbf{1}_v$ in $B$. 
Furthermore, the Laplacian matrix $L$ of $G$ equals $B^{\top} B$. If a row of $B$ corresponds to edge $(u,v)$, its leverage score $(\mathbf{1}_u-\mathbf{1}_v)^{\top} \cdot (B^\top B)^{\dagger} \cdot (\mathbf{1}_u-\mathbf{1}_v)=(\mathbf{1}_u-\mathbf{1}_v)^{\top} \cdot L^{\dagger} \cdot (\mathbf{1}_u-\mathbf{1}_v)$ turns out to be the effective resistance of $(u,v)$ in $G$.

Now we state Lemma~\ref{lemma:laplacian_eigenvalues_perturbation} for the general problem. So Lemma~\ref{lem:removing_edge_eff_res} is a direct corollary of Lemma~\ref{lemma:laplacian_eigenvalues_perturbation} by setting $A$ to the incidence matrix of $G$ as discussed above. 

\begin{restatable}[]{lemma}{EigenvaluePerturbation} \label{lemma:laplacian_eigenvalues_perturbation}
    Given a matrix $A \in \mathbb{R}^{m \times n}$, let $\lambda_1 \le \cdots \le \lambda_n$ be the eigenvalues of $A^{\top} A$. Moreover, for each $\ell \in [m]$, let $\mathbf{a}_{\ell}$ be row $\ell$ of $A$ and  $\tau_{\ell} = \mathbf{a}_\ell^{\top} (A^{\top} A)^{\dagger} \mathbf{a}_\ell$ be its leverage score. 
    
    For any $k$ distinct indices $\ell_1,\ldots,\ell_k \in [m]$, let $A' \in \mathbb{R}^{(m-k) \times n}$ be the matrix obtained from $A$ by removing the corresponding $k$ rows $\mathbf{a}_{\ell_1},\ldots,\mathbf{a}_{\ell_k}$; and let $\lambda'_1 \le \cdots \le \lambda'_n$ be the eigenvalues of  $(A')^\top A'$. It holds that 
    \[
        \forall i \in [n], \lambda'_i \in \big[ (1-\tau_{\ell_1}-\tau_{\ell_2}-\cdots - \tau_{\ell_k}) \cdot \lambda_i, \lambda_i \big].
    \]   
\end{restatable}

\subsection{Related Work}

Previous research has studied the problem of how to quickly compute and approximate the effective resistances in the regime of polynomial-time algorithms, as such algorithms can be used as a crucial subroutine for other graph algorithms. For example, for any two vertices $s$ and $t$ in a  $n$-vertex $m$-edge graph, one can $(1+\epsilon)$-approximate the $s$-$t$ effective resistance in  $\tilde{O}(m+n\epsilon^{-2})$~\cite{DurfeeKPRS16} and $\tilde{O}(m \log (1/\epsilon))$~\cite{CohenKMPPRX14} time, respectively. To $(1\pm \epsilon)$-approximate the effective resistances between $s$ given pairs, Chu et al.~\cite{CGPSSW_pairs_ER_short_cycle} provided an algorithm in time $O(m^{1+o(1)} +(n+s)n^{o(1)}\cdot \epsilon^{-1.5})$. There are also algorithms that find $(1+\epsilon)$-approximations to the effective resistance between every pair of vertices in $\tilde{O}(n^2/\epsilon)$ time~\cite{JS18:sketch}. In order to exactly compute the $s$-$t$ or all-pairs effective resistance, the current fastest algorithms run in times $O(n^\omega)$ (by using the fastest matrix inversion algorithm~\cite{bunch1974,ibarra1982}),
where $\omega < 2.373$ is the matrix multiplication exponent~\cite{alman2021refined}. Faster algorithms are known for planar graphs by using the nested dissection method \cite{LiptonRT79}. 

There exists a line of works on how to efficiently maintain the effective resistances \emph{dynamically} \cite{San04:dynamic,AbrahamDKKP16,GHP17:sparsifier,goranci2018dynamic,DGGP18:dynER,ChenGHPS20}, i.e., if the graph undergoes edge insertions and/or deletions, and the goal is to support the update operations and query for the effective resistances as quickly as possible, rather than having to recompute it from scratch each time.


For the total effective resistance, Ghosh et al. \cite{ghosh2008minimizing} studied algorithms for  allocating edge weights on a given graph in order to minimize the total effective resistance. Li and Zhang \cite{LZ18:kirchhoff} used Kirchhoff index (i.e., total effective resistance) as the measure of edge centrality in weighted
networks and gave efficient algorithms for the measure. 

Matrix perturbation considers the eigenvalues (singular values) of $A$ after adding a matrix $E$ of the same order. Various bounds on the error of matrix perturbation, such as absolute errors and relative errors, has been studied. We refer to the survey \cite{ipsen_1998} and the reference therein for a complete overview. Even though Lemma~\ref{lemma:laplacian_eigenvalues_perturbation} is an instantiation of Theorem 2.8 in \cite{ipsen_1998} with (statistical) leverage scores, we provide two simpler proofs in this work which are more intuitive.
Also, different perturbation bounds in terms of the leverage scores were provided in \cite{ipsensensitivity, holodnak2015conditioning}. 

Leverage scores, analogue to effective resistances of graphs, have wide applications in randomized matrix algorithms and large-scale data algorithms. The most notable property is that sampling the rows of a matrix $A$ via its leverage scores gives an efficient construction of the subspace embedding of $A$ \cite{SpielmanS11}. This fact is extremely useful in designing ultra-efficient algorithms for linear regression and low rank approximations \cite{ClarksonW13}. We refer to the survey \cite{Woodruff14} for a list of applications. 
Since bounding eigenvalues are sufficient for $\ell_2$-subspace embeddings, there is a line of research on the connection between eigenvalues and leverage scores, such as spectral sparsifications \cite{SpielmanS11,BatsonSS12}. 

Sublinear-time algorithms for the related graph problems has also been investigated. For example, Lee \cite{lee2013probabilistic} gave an algorithm for producing a probabilistic $(\epsilon,\delta)$-spectral sparsifier with $O(n\log n/\epsilon^2)$ edges in $\tilde{O}(n/\epsilon^2\delta)$ time for unweighted undirected graph. Note that its running time is sublinear if the number of edges in the graph is large enough. For spectral approximations in sublinear time, various approximation guarantees have been studied in \citep{CKSV18,MNSUW18_spectral_approx,BKM22}.



\vspace{.1in}
\noindent\textbf{Organization.} In Section~\ref{sec:preli}, we provide the basic definitions and notations of this work. Next we discuss about eigenvalues of perturbed matrices and graphs in Section~\ref{sec:proof_eigenvalues_perturbation}. Then we prove the lower bounds of Theorem~\ref{thm:lower_bound_small_degree} in Section~\ref{sec:lower_bound_deg_3}. Due to the space constraint, we defer the proof of Theorem~\ref{thm:lower_bound_large_degree} to Appendix~\ref{sec:large_constant_degree}. Then the approximation algorithm for degree-$2$ graphs is shown in Appendix~\ref{sec:approx_alg}. Finally, we discuss Theorem~\ref{thm:total_tight_lower_bound} about total effective resistances in Appendix~\ref{sec:total_eff_res}.

\section{Preliminaries}\label{sec:preli}
For any integer $k\geq 1$, let $[k]:=\{1,\cdots,k\}$. We use $a=b \pm c$ to denote $a \in [b-c, b+c]$ and $\mathbf{1}$ to denote the all 1 vector.

\vspace{.1in}
\noindent\textbf{Basic definitions from graph theory.} In this work, we only consider undirected graphs with unit weights on each edge. Given an undirected graph $G = (V, E)$ with $n:=|V|$ vertices and $m:=|E|$ edges, let $A_G \in \mathbb{R}^{n \times n}$ denote the adjacency matrix of $G$ and $D_G \in \mathbb{R}^{n \times n}$ denote its degree diagonal matrix. We use $L_G \in \mathbb{R}^{n \times n}$ to denote its Laplacian, i.e., $L_G = D_G - A_G$. Also, we use $V(G)$ and $E(G)$ to denote the vertex set and edge set of a graph $G$.

In this work, we use $\tilde{L}_G$ to denote the normalized Laplacian of $G$, i.e., $\tilde{L}_G:=D_G^{-1/2} \cdot L_G \cdot D_G^{-1/2}$. When the graph is clear, we hide the notation $G$. Also, we use $V(G)$ and $E(G)$ to denote the vertex and edge set of $G$. Moreover, we use $d$ to denote the maximum degree of $G$. For a vertex $u$, let $d(u)$ denote its degree in $G$ and $\mathbf{1}_u \in \mathbb{R}^n$ denote the indicator vector of $u$, i.e., $\mathbf{1}_u(v)=1$ if $v=u$ and 0 otherwise. 

Now we define the adjacency list model for sublinear time graph algorithms \cite{ron2019sublinear}. There are three types of operations in constant time:
\begin{enumerate}
    \item degree query: the algorithm queries the degree of a fixed vertex $v \in V$;
    \item neighbor query: the algorithm queries the $i$-th neighbor of vertex $v$ given $v$ and $i$;
    \item uniform sampling: the algorithm receives a random vertex in $V$.
\end{enumerate}

\noindent\textbf{Basic definitions about matrices and expander graphs.} We use $\psi$, $\phi$, and bold letter $\mathbf{a}$ to denote vectors and $\|\cdot\|$ to denote their $L_2$ (Euclidean) norms. For a vector $\mathbf{a} \in \mathbb{R}^V$ and a subset $U \subset V$, let $\mathbf{a}(U)$ denote the vector in $\mathbb{R}^{U}$ which contains the corresponding entries in $U$. So $\mathbf{a}(i)$ denotes the $i$-th entry of $\mathbf{a}$.

Given a symmetric matrix $A$, we always use $\lambda_1(A) \le \lambda_2(A) \le \cdots \le \lambda_n(A)$ to denote its eigenvalues in the non-decreasing order. Furthermore, let its eigendecomposition be $A=\sum_i \lambda_i(A) \cdot \psi_i \psi_i^T $ where $\psi_i$ is the eigenvector corresponding to the eigenvalue $\lambda_i(A)$.

We say $G$ is an expander if the second smallest eigenvalue $\lambda_2(\tilde{L})$ of $\tilde{L}$ is at least $c_1$, for some universal $c_1>0$. This is equivalent to $\lambda_2(L_G) \ge c_2$ for some $c_2>0$ when the degree of $G$ is bounded. 
We will use Ramanujan graphs \cite{margulis1988explicit,MORGENSTERN199444} of degree 3 and near-Ramanujan graphs for every degree $d \ge 4$ \cite{MODP_22_every_degree}. The guarantee of a Ramanujan graph $G$ of regular degree $d$ is that $\lambda_2(L_G) \ge d - 2 \sqrt{d-1}$. For near-Ramanujan graphs, we only need $\lambda_2(L_G) \ge d - 2.01 \sqrt{d-1}$.

\vspace{.1in}
\noindent\textbf{Basic definitions about effective resistances.}
Then we define the Moore-Penrose pseudo-inverse and effective resistances. For a symmetric matrix $M \in \mathbb{R}^{n \times n}$ whose eigendecomposition is $M = \sum_{i} \lambda_{i} \psi_i \psi_i^T$, its Moore-Penrose pseudo-inverse $M^{\dagger} = \sum_{i:\lambda_i \ne 0} \frac{1}{\lambda_i} \psi_i \psi_i^T$.

\begin{definition}[Effective Resistances] Given a graph $G=(V,E)$, for any two vertices $s,t \in V$, the $s-t$ effective resistance is defined as $R_G(s,t):=(\mathbf{1}_s-\mathbf{1}_t)^{\top} \cdot L_G^{\dagger} \cdot (\mathbf{1}_s-\mathbf{1}_t)$. Moreover, the total effective resistance of $G$ is defined as $R_{tot}(G) = \sum_{i < j} R_G(i,j)$.
\end{definition}
In this work, we will extensively use the following facts about effective resistances (see \cite{lovasz1993random,Spielman} for their proofs). 
\begin{lemma}\label{lem:basic_facts}
Given a graph $G=(V,E)$, the effective resistances in $G$ satisfy the following properties:
\begin{enumerate}
    \item $\sum_{(u,v) \in E} R_G(u,v)=n-1$ and $R_{tot}(G)=n \cdot \sum \limits_{i=2,\ldots,n} \frac 1 {\lambda_i(L_G)}$. 
    \item $2m \cdot R_G(u,v)=\kappa_G(u,v)$, where $\kappa_G(i,j)$ is the \emph{commute time} of a simple random walk from vertex $i$ to $j$ in $G$, i.e., the expected number of steps in a random walk starting at $i$, after vertex  $j$ is visited and then vertex $i$ is reached again.
    \item $\frac{1}{2} \big(\frac{1}{d(u)}+\frac{1}{d(v)}\big) \le R_G(u,v) \le \big( 1/\lambda_2(\widetilde{L}_G) \big) \cdot \big( \frac{1}{d(u)}+\frac{1}{d(v)} \big)$ where $\lambda_2(\widetilde{L}_G)$ is the 2nd smallest eigenvalue of the normalized Laplacian $\widetilde{L}_G$. 
    \item Given any $(s,t)$, consider all functions 
    $\phi \in \mathbb{R}^{V}$ such that $\phi(s)=1$ and $\phi(t)=0$, then
    \[
    R_{G}(s,t)= \frac{1}{\underset{\phi \in \mathbb{R}^{V}: \phi(s)=1, \phi(t)=0}{\min} \sum_{(u,v) \in E} (\phi(u)-\phi(v))^2}.
    \]
    In fact, the minimum value above is acheived when $\phi$ is harmonic, namely the unique solution satisfying $\phi(s)=1, \phi(t)=0,$ and $\phi(v)=\frac{1}{d(v)} \sum_{(u,v)\in E} \phi(u)$ for $v \in V \setminus \{s,t\}$.
 
 An equivalent definition is  $R_G(s,t)=\underset{\phi \in \mathbb{R}^{n}: \phi \bot \mathbf{1}}{\max} \frac{\langle \mathbf{1}_s-\mathbf{1}_t, \phi \rangle^2}{\phi^{\top} \cdot L_G \cdot \phi}$. In particular, the leverage score $\tau_\ell$ of row $\mathbf{a}_\ell$ in $A$ also equals
    $\tau_\ell=\underset{\phi \in \mathbb{R}^{n}}{\max} \frac{\langle \mathbf{a}_\ell, \phi \rangle^2}{\|A \cdot \phi\|_2^2} $.
    
    \item Let $\mathcal{T}(G)$ denote all spanning trees in $G$. Let us pick $T 
    \in \mathcal{T}$ uniformly at random. Then $R_G(u,v)$ is the probability $(u,v)$ is in $T$, i.e., $R_G(u,v) = \Pr_{T \sim \mathcal{T}(G)}[(u,v) \in T]$. 
\end{enumerate}
\end{lemma}
Note that the first equation $\sum_{(u,v) \in E} R_G(u,v)=n-1$ in Lemma \ref{lem:basic_facts} considers the summation over all edges in $E$ where the total effective resistance is the summation over all pairs.

\section{Eigenvalues of Perturbed Graphs and Matrices}\label{sec:proof_eigenvalues_perturbation}
We discuss Lemma~\ref{lem:removing_edge_eff_res} and Lemma~\ref{lemma:laplacian_eigenvalues_perturbation}  in this section. We give two different proofs for Lemma~\ref{lemma:laplacian_eigenvalues_perturbation} such that Lemma~\ref{lem:removing_edge_eff_res} is a direct corollary by setting $A$ to be the incidence matrix. Then we show expander graphs with edges of effective resistance $1$ to illustrate that we have to remove edges with $R_G(u,v)<1$ to keep the perturbed graph as an expander in Lemma~\ref{lem:removing_edge_eff_res}.

We restate Lemma~\ref{lemma:laplacian_eigenvalues_perturbation} again. One remark is that both bounds could be tight in Lemma~\ref{lemma:laplacian_eigenvalues_perturbation}. For example, consider the case $A^\top A=I$ whose $\tau_{\ell}=\|\mathbf{a}_\ell\|_2^2$ for all $\ell$ and $\lambda_i \equiv 1$ for all $i$. Then after removing any $\mathbf{a}_\ell$ (hence $k=1$), $(A')^\top A'=I - \mathbf{a}_\ell \cdot \mathbf{a}_\ell^\top$ has $\lambda'_1=1-\|\mathbf{a}_\ell\|_2^2$ with eigenvector $\frac{\mathbf{a}_\ell}{\|\mathbf{a}_\ell\|_2}$. So $\lambda'_1$ matches the lower bound $(1-\tau_{\ell}) \cdot \lambda_1$; and all the rest eigenvalues of $(A')^\top A'$ are $1$, the same as those of $A^\top A$.

\EigenvaluePerturbation*

The first proof is based on the characteristic polynomial of $A^{\top} A$, which is motivated by the potential function method of classical work \cite{BatsonSS12}. One ingredient in this proof is the matrix determinant lemma \cite{BatsonSS12}.

\begin{lemma}\label{lem:matrix_determinant}
    If $A$ is nonsingular and $v$ is a vector, then $
        det(A+v v^{\top})=\det(A) \cdot (1+v^{\top} A^{-1} v)$.
\end{lemma}



\begin{proofof}{Lemma~\ref{lemma:laplacian_eigenvalues_perturbation}}

For ease of exposition, we start with $k=1$. Namely, we remove one row $\mathbf{a}_{\ell}$ from $A$ to obtain $A'$. 
Consider the characteristic polynomial of $(A')^\top A'$:
  \begin{align*}
      det(xI - (A')^\top A') 
      & = det(xI - A^{\top} A + \mathbf{a}_{\ell} \mathbf{a}_{\ell}^{\top}) \\
      & = det \left( (xI - A^{\top} A) \cdot \left( I + (xI - A^{\top} A)^{\dagger} \mathbf{a}_{\ell} \mathbf{a}_{\ell}^{\top} \right) \right) \\
      & = det \left( xI - A^{\top} A \right) \cdot det \left(  I + (xI - A^{\top} A)^{\dagger} \mathbf{a}_{\ell} \mathbf{a}_{\ell}^{\top}  \right).
  \end{align*}

  Let $\psi_i$ be the corresponding eigenvector of $\lambda_i$ in $A$ such that
  $
  (xI - A^{\top} A)^{\dagger} \mathbf{a}_{\ell} \mathbf{a}_{\ell}^{\top} = \left( \sum \limits_{i\in [n]: \lambda_i \neq x} \frac{1}{x - \lambda_i} \psi_i \psi_i^{\top} \right) \mathbf{a}_{\ell} \mathbf{a}_{\ell}^{\top}$. Then we apply the matrix determinant lemma (Lemma~\ref{lem:matrix_determinant}) to $det \left(  I + (xI - A^{\top} A)^{\dagger} \mathbf{a}_{\ell} \mathbf{a}_{\ell}^{\top}  \right)$ such that
  \begin{align*}
    det(xI - (A')^\top A') & = det(xI - A^\top A) \cdot \left( 1 + \mathbf{a}_{\ell}^{\top} \cdot \sum\limits_{i=1}^{n} \frac{1}{x - \lambda_i} \psi_i \psi_i^{\top} \cdot \mathbf{a}_{\ell} \right)\\
    & = det(xI-A^\top A) \cdot \left( 1 - \sum\limits_{i=1}^{n} \frac{1}{\lambda_i - x} \langle \psi_i, \mathbf{a}_\ell \rangle^2 \right). 
  \end{align*}

  Let $\lambda'_1, \lambda'_2, \cdots, \lambda'_n$ be the roots of $det(xI - (A')^\top A')$. Similar to the argument in~\cite{BatsonSS12}: If $\langle \psi_i, \mathbf{a}_\ell \rangle = 0$, then $\lambda_i$ is a root of $det(xI-A^\top A)$, i.e., $\lambda'_i = \lambda_i$. Otherwise, $\lambda'_i \in (\lambda_{i-1},\lambda_i)$ satisfies $ \sum\limits_{j=1}^{n} \frac{1}{\lambda_j - \lambda'_i} \langle \psi_, \mathbf{a}_\ell \rangle^2 = 1$. This is because $\underset{x \rightarrow \lambda_{i-1} + }{\lim} \sum\limits_{i=1}^{n} \frac{1}{\lambda_i - x} \langle \psi_i, \mathbf{a}_\ell \rangle^2 = -\infty$ and $\underset{x \rightarrow \lambda_{i} - }{\lim} \sum\limits_{i=1}^{n} \frac{1}{\lambda_i - x} \langle \psi_i, \mathbf{a}_\ell \rangle^2 = +\infty$. For the 2nd case, we show $\lambda'_i \ge (1-\tau_{\ell}) \lambda_i$.

  Let the function $p(x) := \sum\limits_{j=1}^{n} \frac{1}{\lambda_j - x} \langle \psi_j, \mathbf{a}_\ell \rangle^2$. If $\lambda_{i-1} \ge (1-\tau_{\ell})\lambda_i$, then we have proved $\lambda'_i > (1-\tau_{\ell})\lambda_i$. Otherwise we show $\lambda'_i > (1-\tau_{\ell})\lambda_i$ by considering $p(x)$ in the continuous interval $\big[ (1 - \tau_\ell)\lambda_i, \lambda_i \big)$.
  
  \begin{align*}
      p \bigg( (1 - \tau_\ell)\lambda_i \bigg) 
      & = \sum\limits_{j=1}^{n} \frac{1}{\lambda_j - (1 - \tau_\ell)\lambda_i} \langle \psi_, \mathbf{a}_\ell \rangle^2 \\
      & \le \sum\limits_{j=i}^{n} \frac{1}{\lambda_j - (1 - \tau_\ell)\lambda_i} \langle \psi_, \mathbf{a}_\ell \rangle^2 \tag{Since $\lambda_1<\cdots<\lambda_{i-1}<(1-\tau_\ell)\lambda_i$ from the assumption, their corresponding terms are negative.} \\ 
      & \le \sum\limits_{j=i}^{n} \frac{1}{\lambda_j - (1 - \tau_\ell)\lambda_j} \langle \psi_, \mathbf{a}_\ell \rangle^2 = \frac{1}{\tau_\ell} \sum\limits_{j=i}^{n} \frac{1}{\lambda_j} \langle \psi_, \mathbf{a}_\ell \rangle^2
  \end{align*}
  
  From the definition, $\tau_\ell = \mathbf{a}_\ell^{\top} (A^{\top} A)^{\dagger} \mathbf{a}_\ell = \mathbf{a}_\ell^{\top} \left(\sum\limits_{i=1}^{n} \frac{1}{\lambda_i} \psi_i \psi_i^{\top} \right) \mathbf{a}_\ell = \sum\limits_{i=1}^{n} \frac{1}{\lambda_i} \langle \psi_i, \mathbf{a}_\ell \rangle^2$. So $p \big( (1 - \tau_\ell)\lambda_i \big) \le 1$. 
  
  On the other hand, $\underset{x \rightarrow \lambda_{i}-}{\lim} \sum\limits_{i=1}^{n} \frac{1}{\lambda_i - x} \langle \psi_i, \mathbf{a}_\ell \rangle^2 = +\infty$. So
  $p \big( (1 - \tau_\ell)\lambda_i \big) \le 1$ and $p(\lambda_i - \epsilon) > 1$ infer that there exists a $x \in [(1 - \tau_\ell)\lambda_i, \lambda_i)$ such that $p(x) = 1$, which also means that $det(xI - (A')^\top A')$ has a root $\lambda'_i \in [(1 - \tau_\ell)\lambda_i, \lambda_i)$.

  Next we prove $\lambda_i' \ge (1-\tau_{\ell_1}-\cdots-\tau_{\ell_k}) \lambda_i$ by induction on $k$. The above calculation proves the base case of $k=1$.

  For the inductive step, let $\tilde{A}$ denote the matrix after removing $\mathbf{a}_{\ell_1},\ldots,\mathbf{a}_{\ell_q}$ and $A'$ denote the matrix by removing one more edge $a_{\ell_{q+1}}$. By the induction hypothesis, $\lambda_i(\tilde{A}) \in \left[ \big( 1 - \sum\limits_{j=1}^{q} \tau_{\ell_j} \big) \cdot \lambda_i, \lambda_i \right]$ and $( 1 - \sum\limits_{j=1}^{q} \tau_{\ell_j} ) \cdot A^\top A \preceq \tilde{A}^{\top} \tilde{A}$. This implies $(\tilde{A}^{\top} \tilde{A})^{\dagger} \preceq ( 1 - \sum\limits_{j=1}^{q} \tau_{\ell_j} )^{-1} \cdot (A^\top A  )^{\dagger}$,
  
  \[ \tilde{\tau}_{\ell_{q+1}} = \mathbf{a}_{\ell_{q+1}}^{\top} (\tilde{A}^{\top} \tilde{A})^{\dagger} \mathbf{a}_{\ell_{q+1}} \le \left( 1 - \sum\limits_{j=1}^{q} \tau_{\ell_j} \right)^{-1} \mathbf{a}_{\ell_{q+1}}^{\top} ({A}^{\top} {A})^{\dagger} \mathbf{a}_{\ell_{q+1}} = \left( 1 - \sum\limits_{j=1}^{q} \tau_{\ell_j} \right)^{-1} \tau_{\ell_{q+1}}. \] 
  
  Using the perturbation bound for $k=1$ on $\tilde{A}$ and $A'$, 
  \begin{align*}
      \lambda_i(A') & \in \left[ (1-\tilde{\tau}_{\ell_{q+1}}) \lambda_i(\tilde{A}), \lambda_i(\tilde{A}) \right]\\
      & \subseteq \left[ \left( 1 -  \left( 1 - \sum\limits_{j=1}^{q} \tau_{\ell_j} \right)^{-1} \tau_{\ell_{q+1}} \right) \left( 1 - \sum\limits_{j=1}^{q} \tau_{\ell_j} \right) \lambda_i, \lambda_i \right] = \left[ \left( 1 - \sum\limits_{j=1}^{q+1} \tau_{\ell_j} \right) \lambda_i, \lambda_i \right].
  \end{align*}
\end{proofof}

Next we present the 2nd proof, which is based on Property 4 of the leverage score in Lemma~\ref{lem:basic_facts}. One advantage of this proof is that it works directly on multiple edges.

\begin{proofof}{Lemma~\ref{lemma:laplacian_eigenvalues_perturbation}}
  We will apply the Courant-Fischer theorem below, which shows $\lambda_i' = \min\limits_{\begin{subarray}{c} S \subset \mathbb{R}^n, \\ dim(S) = i \end{subarray} } \max\limits_{\phi \in S} \frac{\phi^{\top} (A')^{\top} A' \phi}{\phi^{\top} \phi}$.
  \begin{lemma}[Courant-Fischer-Weyl theorem] \label{lemma:CFW_theorem}
      Let $H$ be an $n \times n$ Hermitian matrix with eigenvalues $\lambda_1(H) \le \lambda_2(H) \le \cdots \le \lambda_n(H)$, then \[ \lambda_k(H) = \min\limits_{\begin{subarray}{c} S \subset \mathbb{R}^n, \\ dim(S) = k \end{subarray} } \max\limits_{x \in S} \frac{x^{\top} H x}{x^{\top} x}   
      \qquad \text{ and } 
      \max\limits_{\begin{subarray}{c} S \subset \mathbb{R}^n, \\ dim(S) = n-k+1 \end{subarray}}  \min\limits_{x \in S} \frac{x^{\top} H x}{x^{\top} x}. \]
  \end{lemma}
  
 We compare the Rayleigh quotient $\frac{\phi^{\top} (A')^{\top} A' \phi}{\phi^{\top} \phi}$ with $\frac{\phi^{\top} A^{\top} A \phi}{\phi^{\top} \phi}$:
  \[ \frac{\phi^{\top} (A')^{\top} A' \phi}{\phi^{\top} \phi} 
  = \frac{\left\| A' \phi \right\|^2_2}{\left\| \phi \right\|^2_2} 
  = \frac{\left\| A \phi \right\|^2_2 - \sum\limits_{j=1}^{k} \left( a_{\ell_j}^{\top} \phi \right)^2}{\left\| \phi \right\|^2_2} 
  = \frac{\left\| A \phi \right\|^2_2}{\left\| \phi \right\|^2_2} \cdot \left[ 1 - \sum\limits_{j=1}^{k} \frac{ \left( a_{\ell_j}^{\top} \phi \right)^2}{\left\| A \phi \right\|^2_2} \right]. \] 
Since $\tau_{\ell} = \max\limits_{\phi \in \mathbb{R}^n } \frac{ \left( a_{\ell}^{\top} \phi \right)^2}{\left\| A \phi \right\|^2_2}$ from Property 4 of Lemma~\ref{lem:basic_facts}, 
\[
1 - \sum\limits_{j=1}^{k} \frac{ \left( a_{\ell_j}^{\top} \phi \right)^2}{\left\| A \phi \right\|^2_2} \ge 1 - \sum\limits_{j=1}^{k} \tau_{\ell_j}.\]
Using Lemma~\ref{lemma:CFW_theorem} again, 
  \[ \lambda_i' 
  = \min\limits_{\begin{subarray}{c} S \subset \mathbb{R}^n, \\ dim(S) = \phi \end{subarray} } \max\limits_{\phi \in S} \frac{\left\| A \phi \right\|^2_2}{\left\| \phi \right\|^2_2} \cdot \left[ 1 - \sum\limits_{j=1}^{k} \frac{ \left( a_{\ell_j}^{\top} \phi \right)^2}{\left\| A \phi \right\|^2_2} \right] 
  \ge \left(1 - \sum\limits_{j=1}^{k} \tau_{\ell_j} \right) \min\limits_{\begin{subarray}{c} S \subset \mathbb{R}^n, \\ dim(S) = \phi \end{subarray} } \max\limits_{\phi \in S} \frac{\left\| A \phi \right\|^2_2}{\left\| \phi \right\|^2_2} 
  = (1 - \sum\limits_{j=1}^{k} \tau_{\ell_j})\lambda_i.  \]
  
  As $1 - \sum\limits_{j=1}^{k} \frac{ \left( a_{\ell_j}^{\top} \phi \right)^2}{\left\| A \phi \right\|^2_2} \le 1$,  we can get $\lambda_i' \le \lambda_i$ in a similar way. Combining the two inequalities, $\lambda'_i \in [(1 - \sum\limits_{j=1}^{k} \tau_{\ell_j})\lambda_i, \lambda_i]$.
\end{proofof}

We remark that Lemma~\ref{lemma:laplacian_eigenvalues_perturbation} is also an instantiation of Theorem 2.8 in \cite{ipsen_1998} with leverage scores.
However, we believe the above two proofs shed more insights on the structure of perturbed matrices and are simplier (without Weyl's interlacing inequality). But for completeness, we provide that proof in  Appendix~\ref{append}.

Next we show that there exist edges in expander graphs with unit effective resistance.
\begin{claim}\label{clm:expander_unit_effect_resis}
For any $c_0>0$ and infinitely many $n$, there exists an expander graph $G$ of constant degree such that
\begin{enumerate}
    \item The smallest non-trivial eigenvalue of its Laplacian is at least $c_0$;
    \item There exists an edge $e$ in $G$ with effective resistance 1. 
\end{enumerate}
Hence, after removing $e$, $\lambda_2(L_{G'})=0$ in the perturbed graph $G'$ such that $G'$ is no longer an expander.
\end{claim}


\begin{proof}
Given an expander $G'=(V',E')$ with constant degree and $\lambda_2(L_G)=\Omega(1)$, we add an extra vertex of degree $1$ to it. More formally, let $V'=\{1,\ldots,n-1\}$ and $G = (V, E) = (V' \cup \{ n \}, E' \cup \{(n-1,n)\} )$. Then $R_G(n-1,n) = 1$. We will show $\lambda_2\big( L_{G'} \big)=\Omega(1)$ such that $G'$ is also an expander (its degree is still a constant).     So $L_G(i,j)=L_{G'}(i,j)$ except $i \in \{n-1,n\}$ and $j \in \{n-1,n\}$ where those four entries have $L_G(n-1,n-1)=L_{G'}(n-1,n-1)+1, L_G(n-1,n)=L_G(n,n-1) = -1, $ and $L_G(n,n)=1$.
  

    By the Courant-Fischer Theorem, $\lambda_2(L_G) = \min\limits_{\phi \bot \mathbf{1}, \norm{\phi}_2 = 1} \phi^{\top} L_G \phi$. Then we prove $ \phi^{\top} L_G \phi = \Omega(1)$ for all $\phi$ with $\phi \bot \mathbf{1}$ and $\norm{\phi}_2 = 1$ by considering $\phi(n)$ in two cases.

    If $|\phi(n)| \ge 0.8$, then $|\phi(n-1)| \le 0.6$ because $\|\phi\|_2=1$. So $\phi^{\top} L_G \phi = \sum\limits_{(u,v) \in E}(\phi(u)-\phi(v))^2 \ge (\phi(n) - \phi(n-1))^2 \ge 0.04$.
    
    Otherwise $|\phi(n)| < 0.8$. For convenience, we assume $\phi(n) \in [0,0.8)$. Let $\phi([n-1]) \in \mathbb{R}^{[n-1]}$ denotes the sub-vector on $V'$ and $\mathbf{1}([n-1])$ denotes the all $1$-vector on $V'$. Then $\norm{\phi([n-1])} \ge 0.6$ and \[
    \phi^{\top} L_G \phi = \sum\limits_{(u,v) \in E}(\phi(u)-\phi(v))^2 \ge \sum\limits_{(u,v) \in E'} (\phi(u)-\phi(v))^2=\phi([n-1])^\top \cdot L_{G'} \cdot \phi([n-1]).\]

    Since $L_{G'}$ has an eigenvalue $0$ with eigenvector $\mathbf{1}([n-1])$, we consider $\phi([n-1])$ after orthonormalization: $\phi([n-1]) - \frac{\langle \phi([n-1]), \mathbf{1}([n-1]) \rangle}{n-1} \cdot \mathbf{1}([n-1])$. 
    
Note that $\langle \phi, \mathbf{1} \rangle = 0$ implies $\langle \phi([n-1]), \mathbf{1}([n-1]) \rangle = - \phi(n)$. We calculate its $L_2$ norm after orthonormalization as
    \begin{align*}
    & \left\| \phi([n-1]) - \frac{\langle \phi([n-1]), \mathbf{1}([n-1]) \rangle}{n-1} \cdot \mathbf{1}([n-1]) \right\|^2 \\
    = & \left\|\phi([n-1]) + \frac{\phi(n)}{n-1} \cdot \mathbf{1}([n-1]) \right\|^2 \\
    = & \left\|\phi([n-1])\right\|^2 + \left\|\frac{\phi(n)}{n-1} \cdot \mathbf{1}([n-1]) \right\|^2 + \frac{2\phi(n)}{n-1} \cdot \langle \phi([n-1]), \mathbf{1}([n-1]) \rangle  \\
    \ge & \left\|\phi([n-1])\right\|^2 - \frac{2\phi(n)}{n-1} \cdot \phi(n) \ge 0.36 - \frac{2}{n-1} \ge 0.3.     
    \end{align*}
    So $\phi^{\top} L_G \phi \ge \lambda_2(L_{G'}) \cdot \left\|\phi([n-1]) - \frac{\langle \phi([n-1]), \mathbf{1}([n-1]) \rangle}{n-1} \cdot \mathbf{1}([n-1]) \right\|^2 = \Omega(1)$.

   From all discussion above, $\lambda_2(L_G) = \Omega(1)$. Since $G$ has a constant degree, $\lambda_2(\tilde L_G) \ge  \lambda_2(L_G) / d(G) = \Omega(1)$, which means $G$ is also an expander.
\end{proof}

\section{Lower Bound for Degree 3}\label{sec:lower_bound_deg_3}
In this section, we prove the lower bound for graphs with degrees at most $3$ (except the given pair $s$ and $t$). Recall the statement of Theorem~\ref{thm:lower_bound_small_degree} in Section~\ref{sec:intro}.

\lowerboundsmalldegree*

One ingredient in the proof is Lemma~\ref{lem:removing_edge_eff_res}, which bounds the expansion of a perturbed graph $G'$ obtained from removing one edge $e$ of the original graph $G$, in terms of the eigenvalues of the Laplacian of $G$ and the effective resistance $R_G(e)$.

In the rest of this section, we finish the proof of Theorem~\ref{thm:lower_bound_small_degree}. The high level idea is to consider a graph $G$ of degree 3 (see Figure~\ref{ex:dumbshell} for an illustration) constituted by two disjoint expanders $H_s$ and $H_t$ with one extra edge between $s$ and $t$ in these two expanders separately. 

Then we produce a random graph $G'$ as follows (see Figure~\ref{ex:modification} for an illustration): remove one random edge $(u,u')$ in $H_s$ and another one $(v,v')$ in $H_t$ separately; then add two edges $(u,v)$ and $(u',v')$ to $G'$. 
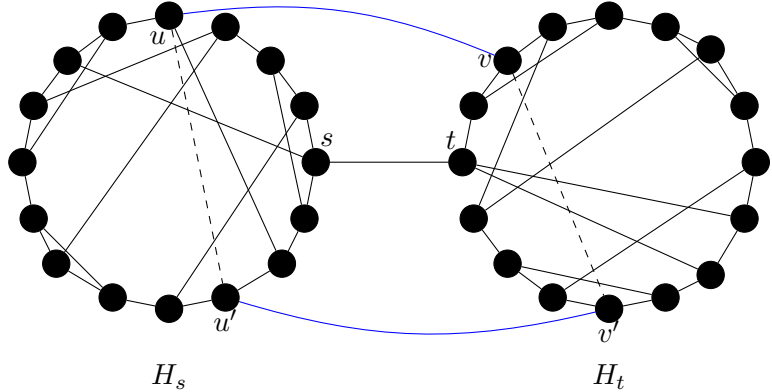
\begin{figure}[h]
    \centering
    \caption{$G'$ after modification}\label{ex:modification}
    \begin{tikzpicture}[scale = 0.6]
        \tikzstyle{vertex}=[fill=black, draw=black, shape=circle]
		\node [style=vertex] (0) at (-0.5, 3.5) {};
		\node [style=vertex] (2) at (1.25, 5.25) {};
		\node [style=vertex] (3) at (4.75, 4.5) {};
		\node [style=vertex] (4) at (0.25, 4.5) {};
		\node [style=vertex] (5) at (3.75, -0.75) {};
		\node [style=vertex] (6) at (-0.75, 2.25) {};
		\node [style=vertex] (7) at (5, 0) {};
		\node [style=vertex] (8) at (2.5, -1) {};
		\node [style=vertex] (9) at (5.5, 3.5) {};
		\node [style=vertex] (10) at (2.5, 5.5) {};
		\node [style=vertex] (11) at (0, 0) {};
		\node [style=vertex] (12) at (5.5, 1) {};
		\node [style=vertex] (13) at (3.75, 5.25) {};
		\node [style=vertex] (14) at (5.75, 2.25) {};
		\node [style=vertex] (15) at (1.25, -0.75) {};
		\node [style=vertex] (16) at (-0.5, 1) {};
		\node [style=vertex] (17) at (11, 5.25) {};
		\node [style=vertex] (18) at (9.25, 3.5) {};
		\node [style=vertex] (19) at (10, 0) {};
		\node [style=vertex] (20) at (10, 4.5) {};
		\node [style=vertex] (21) at (15.25, 1) {};
		\node [style=vertex] (22) at (12.25, 5.5) {};
		\node [style=vertex] (23) at (14.5, -0.25) {};
		\node [style=vertex] (24) at (15.5, 2.25) {};
		\node [style=vertex] (25) at (11, -0.75) {};
		\node [style=vertex] (26) at (9, 2.25) {};
		\node [style=vertex] (27) at (14.5, 4.75) {};
		\node [style=vertex] (28) at (13.5, -0.75) {};
		\node [style=vertex] (29) at (9.25, 1) {};
		\node [style=vertex] (30) at (12.25, -1) {};
		\node [style=vertex] (31) at (15.25, 3.5) {};
		\node [style=vertex] (32) at (13.5, 5.25) {};
		\draw (3) to (9);
		\draw (2) to (4);
		\draw (0) to (4);
		\draw (0) to (6);
		\draw (2) to (6);
		\draw (5) to (7);
		\draw (15) to (8);
		\draw (4) to (14);
		\draw (0) to (13);
		\draw (14) to (12);
		\draw (12) to (7);
		\draw (13) to (11);
		\draw (16) to (11);
		\draw (10) to (7);
		\draw (3) to (13);
		\draw (3) to (12);
		\draw (9) to (8);
		\draw (16) to (15);
		\draw (9) to (14);
		\draw (11) to (15);
		\draw (5) to (8);
		\draw (2) to (10);
		\draw (6) to (16);
		\draw (19) to (25);
		\draw (18) to (20);
		\draw (17) to (20);
		\draw (17) to (22);
		\draw (18) to (22);
		\draw (21) to (23);
		\draw (31) to (24);
		\draw (17) to (29);
		\draw (30) to (28);
		\draw (28) to (23);
		\draw (29) to (27);
		\draw (32) to (27);
		\draw (26) to (21);
		\draw (26) to (23);
		\draw (19) to (29);
		\draw (19) to (28);
		\draw (25) to (24);
		\draw (32) to (31);
		\draw (25) to (30);
		\draw (27) to (31);
		\draw (21) to (24);
		\draw (18) to (26);
		\draw (22) to (32);
		\draw (14) to (26);
		\draw [dashed] (10) to (5);
		\draw [dashed] (20) to (30);
		\draw [draw=blue, bend left=15] (10) to (20);
		\draw [draw=blue, bend right=15] (5) to (30);
        \node at(6, 2.75) {$s$};
        \node at(8.75, 2.75) {$t$};
        \node at(2.25, 5) {$u$};
        \node at(3.75, -1.25) {$u'$};
        \node at(9.5, 4.5)  {$v$};
        \node at(12.25, -1.5) {$v'$};
        \node at(2.5, -2.5)  {$H_s$};
        \node at(12.25, -2.5) {$H_t$};
    \end{tikzpicture}
\end{figure}
So the number of vertices in $G$ and $G'$ are the same; moreover all vertices have the same degrees. Thus the only way to distinguish between $G$ and $G'$ is to figure out whether any of these four edges $(u,u'), (v,v'), (u,v), (u',v')$ is in the graph or not. Since there are $3n/2$ edges, this needs $\Omega(n)$ queries. Finally we bound $R_{G'}(s,t)$ (and the approximation ratio)  via Lemma~\ref{lem:removing_edge_eff_res}.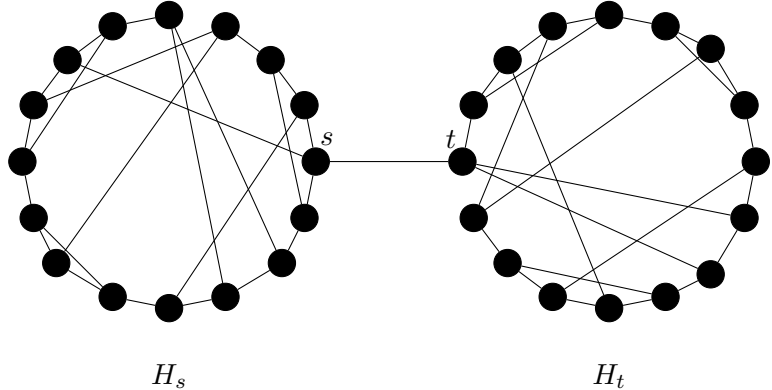
\begin{figure}[h]
    \centering
    \caption{Construction of $G$ with two expanders $H_s$ and $H_t$}\label{ex:dumbshell}
    \begin{tikzpicture}[scale = 0.6]
        \tikzstyle{vertex}=[fill=black, draw=black, shape=circle]
        \node [style=vertex] (0) at (-0.5, 3.5) {};
        \node [style=vertex] (2) at (1.25, 5.25) {};
        \node [style=vertex] (3) at (4.75, 4.5) {};
        \node [style=vertex] (4) at (0.25, 4.5) {};
        \node [style=vertex] (5) at (3.75, -0.75) {};
        \node [style=vertex] (6) at (-0.75, 2.25) {};
        \node [style=vertex] (7) at (5, 0) {};
        \node [style=vertex] (8) at (2.5, -1) {};
        \node [style=vertex] (9) at (5.5, 3.5) {};
        \node [style=vertex] (10) at (2.5, 5.5) {};
        \node [style=vertex] (11) at (0, 0) {};
        \node [style=vertex] (12) at (5.5, 1) {};
        \node [style=vertex] (13) at (3.75, 5.25) {};
        \node [style=vertex] (14) at (5.75, 2.25) {};
        \node [style=vertex] (15) at (1.25, -0.75) {};
        \node [style=vertex] (16) at (-0.5, 1) {};
        \node [style=vertex] (17) at (11, 5.25) {};
        \node [style=vertex] (18) at (9.25, 3.5) {};
        \node [style=vertex] (19) at (10, 0) {};
        \node [style=vertex] (20) at (10, 4.5) {};
        \node [style=vertex] (21) at (15.25, 1) {};
        \node [style=vertex] (22) at (12.25, 5.5) {};
        \node [style=vertex] (23) at (14.5, -0.25) {};
        \node [style=vertex] (24) at (15.5, 2.25) {};
        \node [style=vertex] (25) at (11, -0.75) {};
        \node [style=vertex] (26) at (9, 2.25) {};
        \node [style=vertex] (27) at (14.5, 4.75) {};
        \node [style=vertex] (28) at (13.5, -0.75) {};
        \node [style=vertex] (29) at (9.25, 1) {};
        \node [style=vertex] (30) at (12.25, -1) {};
        \node [style=vertex] (31) at (15.25, 3.5) {};
        \node [style=vertex] (32) at (13.5, 5.25) {};
        \draw (3) to (9);
        \draw (2) to (4);
        \draw (0) to (4);
        \draw (0) to (6);
        \draw (2) to (6);
        \draw (5) to (7);
        \draw (15) to (8);
        \draw (4) to (14);
        \draw (0) to (13);
        \draw (14) to (12);
        \draw (12) to (7);
        \draw (13) to (11);
        \draw (16) to (11);
        \draw (10) to (5);
        \draw (10) to (7);
        \draw (3) to (13);
        \draw (3) to (12);
        \draw (9) to (8);
        \draw (16) to (15);
        \draw (9) to (14);
        \draw (11) to (15);
        \draw (5) to (8);
        \draw (2) to (10);
        \draw (6) to (16);
        \draw (19) to (25);
        \draw (18) to (20);
        \draw (17) to (20);
        \draw (17) to (22);
        \draw (18) to (22);
        \draw (21) to (23);
        \draw (31) to (24);
        \draw (20) to (30);
        \draw (17) to (29);
        \draw (30) to (28);
        \draw (28) to (23);
        \draw (29) to (27);
        \draw (32) to (27);
        \draw (26) to (21);
        \draw (26) to (23);
        \draw (19) to (29);
        \draw (19) to (28);
        \draw (25) to (24);
        \draw (32) to (31);
        \draw (25) to (30);
        \draw (27) to (31);
        \draw (21) to (24);
        \draw (18) to (26);
        \draw (22) to (32);
        \draw (14) to (26);
        \node at(6, 2.75) {$s$};
        \node at(8.75, 2.75) {$t$};
        \node at(2.5, -2.5)  {$H_s$};
        \node at(12.25, -2.5) {$H_t$};
    \end{tikzpicture}
\end{figure}

\begin{proofof}{Theorem~\ref{thm:lower_bound_small_degree}}
  We provide a distribution $\mathcal{G}$ of graphs with $n$ vertices and degree at most $3$ (except $s$ and $t$). Then by Yao's minimax principle (see Lemma~\ref{Yao_minimax_principle}), we only need to consider deterministic algorithms of $q_n$ neighbor queries and $1.01$-approximation ratio whose success probability is at least 0.6 over $\mathcal{G}$. Our goal is to prove $q_n=\Omega(n)$.
  
  Consider any $n$ such that there exists a 3-regular Ramanujan graph $H$ of size $n/2$ \cite{margulis1988explicit,MORGENSTERN199444}. Then we construct $G$ with the given pair $(s,t)$ as follows:
  \begin{enumerate}
      \item Take two vertex-disjoint copies of $H$, denoted by $H_s$ and $H_t$, such that $H_s$ contains vertex $s$ and $H_t$ contains vertex $t$.
      \item Define the vertex set $V(G)$ of $G$ to be the union of the vertex sets of $H_s$ and $H_t$.
      \item Define the edge set $E(G)$ of $G$ to be the union of $(s,t)$, the edge set $E(H_s)$ of $H_s$ and the edge set $E(H_t)$ of $H_t$.
  \end{enumerate}
Note that the effective resistance $R_G(s,t)=1$  in $G$.

  Next, we construct a random graph $G'$ based on $G$. Let $E_{s,3/4}$ be the set of edges in $H_s$ with effective resistance at most $3/4$, i.e., $E_{s,3/4}=\big\{(u,v)\in E(H_s) \big| R_G(u,v) \le 3/4 \big\}$ and we define $E_{t,3/4}$ analogously. We use the following claim to lower bound the sizes of these two sets whose proof is deferred to Section~\ref{sec:omitted_proofs}.
  \begin{claim}\label{claim:number_possible_edges}
It holds that  $|E_{s,3/4}| \ge n/12$ and $|E_{t,3/4}| \ge n/12$.
  \end{claim}
  We give the construction of $G'$ which is almost identical to $G$ except four edges.
  \begin{enumerate}
      \item Choose one edge $(u,u')$ uniformly at random from $E_{s,3/4}$ and remove it from $G$. Similarly, remove another random edge $(v,v') \in E_{t,3/4}$. For convenience, let $H'_s$ be the subgraph of $H_s$ obtained by removing $(u,u')$ from $H_s$ and $H'_t$ be the subgraph obtained by removing $(v,v')$ from $H'_t$.
      \item Add $(u,v)$ and $(u',v')$ to $G'$.
  \end{enumerate}
By our choices of  $(u,u')\in E_{s,3/4}$  and $(v,v') \in E_{t,3/4}$, $H'_s$ and $H'_t$ are still expander graphs from Lemma~\ref{lem:removing_edge_eff_res}. Based on this property, we show the effective resistance between $s$ and $t$ in $G'$ is strictly less than 1 in Claim~\ref{clm:new_effective_resis}, whose proof is deferred to Section~\ref{sec:omitted_proofs}.
  \begin{claim}\label{clm:new_effective_resis}
It holds that      $R_{G'}(s,t) \le 0.99$.
  \end{claim}

  Now we fix a deterministic algorithm $A$ with approximation ratio at most $1.01$ and consider the underlying distribution $\mathcal{G}$, which is $G$ or $G'$ with probability $1/2$ separately. Observe that whenever $A$ succeeds, $A$ is able to distinguish between $G$ and $G'$, since the ratio between $R_G(s,t)=1$ and $R_{G'}(s,t) \le 0.99$ is more than $1.01$. For convenience, let us modify $A$ so that its output is 
 an assertion about whether the input graph is $G$ or $G'$. By Lemma~\ref{Yao_minimax_principle} of Yao's minimax principle, it holds that
  \begin{equation}\label{eq:dist_prob}
      0.6 \le \Pr_{\mathcal{G}}[A \text{ succeeds}]=\frac{\Pr[A(G)=G]}{2} + \frac{\Pr[A(G')=G']}{2}.
  \end{equation}
  Next we consider all neighbor queries made by $A$ when the underlying graph is $G$. Since $A$ and $G$ are fixed, say $A$ makes $q_n$ fixed neighbor queries on $G$. If $G'$ and $G$ provide the same answers on these neighbor queries, $A$ fails to distinguish them. But $G'$ is obtained from $G$ by removing one random edge in $E_{s,3/4}$ and another one in $E_{t,3/4}$ separately. Hence at most $q_n$ edges in $E_{s,3/4}$ will be queried; and similarly for $E_{t,3/4}$. So we bound
  \begin{align*}
  \Pr[A(G')=G'] & \le \Pr[A(G)=G'] + \Pr[\text{One neighbor query returns different values}] \\
  & \le 1-\Pr[A(G)=G] + q_n/|E_{s,3/4}| + q_n/|E_{t,3/4}|.  
  \end{align*}
  Plugging this into \eqref{eq:dist_prob} with the two bounds in Claim~\ref{claim:number_possible_edges}, we obtain $q_n \ge 0.1 \cdot \frac{n}{12}$. 
  

  
\end{proofof}

We remark that replacing $H$ by an expander of degree $[m/n]$ would give a lower bound $\Omega(m)$ instead of $\Omega(n)$, but this result is covered by Theorem~\ref{thm:lower_bound_large_degree} with $\ell=1$.


\subsection{Proofs of Claim~\ref{claim:number_possible_edges} and Claim~\ref{clm:new_effective_resis}}\label{sec:omitted_proofs}

We now give the proof of Claim~\ref{claim:number_possible_edges}. 

\begin{claimproofof}{Claim~\ref{claim:number_possible_edges}}
Recall that $E_{s,3/4}=\big\{(u,v)\in E(H_s) \big| R_G(u,v) \le 3/4 \big\}$. First, for any edge $(u,v)$ in $H_s$, $R_{H_s}(u,v)= R_{G}(u,v)$. Since $(s,t)$ is the unique edge between $H_s$ and $H_t$, the set of spanning trees $\mathcal{T}(G)$ is generated by picking one spanning tree $T_1 \in \mathcal{T}(H_s)$ and one $T_2 \in \mathcal{T}(H_t)$ then connecting them by $\{(s,t)\}$. Then by Property 5 of Lemma~\ref{lem:basic_facts}, $R_{H_s}(u,v)= R_{G}(u,v)$. 
Thus, $E_{s,3/4}=\big\{(u,v)\in E(H_s) \big| R_{H_s}(u,v) \le 3/4 \big\}$. 
    
From Property 1 of Lemma~\ref{lem:basic_facts}, $\sum_{(u,v) \in E(H_s)} R_{H_s}(u,v)=\frac{n}{2}-1$. By the definition of $E_{s,3/4}$, we have $|E(H_s)\setminus E_{s,3/4}|\cdot \frac{3}{4}\leq \sum_{(u,v) \in E(H_s)\setminus E_{s,3/4}} R_{H_s}(u,v)\leq \sum_{(u,v) \in E(H_s)} R_{H_s}(u,v)= \frac{n}{2}-1$, which implies that $|E(H_s)\setminus E_{s,3/4}|\leq \frac{n/2-1}{3/4}$. 

Since $|E(H_s)|=\frac{n}{2} \cdot \frac{3}{2}$, 
it holds that
$|E_{s,3/4}| \ge \frac{n}{2} \cdot \frac{3}{2} - \frac{n/2-1}{3/4} \ge n/12$.
\end{claimproofof}

While it is possible to use the Cayley graph construction of Ramanujan graphs to obtain a better bound, we did not attempt to optimize those constants in this work. Then we finish the proof of Claim~\ref{clm:new_effective_resis}.

\begin{claimproofof}{Claim~\ref{clm:new_effective_resis}}  
  Recall that $H'_s$ and $H'_t$ are subgraphs in $G'$ obtained by removing $(u,u')$ in $H_s$ and $(v,v')$ in $H_t$ separately. We show that the $2$nd eigenvalue of $L_{H_s^{'}}$ has $\lambda_2(L_{H_s^{'}}) \ge \frac{3-2\sqrt{2}}{4}$. As a Ramanujan graph $H_s$, $\lambda_2(L_{H_s}) \ge 3 - 2 \sqrt{2}$. Then we apply Lemma~\ref{lem:removing_edge_eff_res} to $\lambda_2$: since $R_{H_s}(u,u')=R_{G}(u,u') \le 3/4$ (from the proof of Claim~\ref{claim:number_possible_edges}), $\lambda_2(L_{H_s^{'}}) \ge \lambda_2(L_{H_s}) \cdot (1 - R_{H_s}(u,u')) \ge \frac {3 - 2 \sqrt{2}} {4}$.  
  
  
   
   As removing other edges doesn't decrease the effective resistance of $(s,t)$, we ignore $(u', v')$ to give a upper bound of $R_{G'}(s,t)$.  Consider another path  $s-u-v-t$, we have \begin{align*}
      R_{H_s'}(s,u) & = (\mathbf{1}_s - \mathbf{1}_u)^{\top} L_{H_s'}^{\dagger} (\mathbf{1}_s - \mathbf{1}_u) \\
      & \le \lambda_n(L_{H_s^{'}}^{\dagger}) \cdot \norm{\mathbf{1}_s - \mathbf{1}_u}_2^2 = \frac{1}{\lambda_2(L_{H_s^{'}})} \cdot \norm{\mathbf{1}_s - \mathbf{1}_u}_2^2 = \frac{4}{3-2\sqrt{2}} \cdot 2 = 24 + 16 \sqrt{2};
  \end{align*}
  and the same bound holds for $R_{H_t'}(v,t)$. Since $(s,t)$ and the path passing by $s-u-v-t$ are in parallel, we have $ R_{G'}(s,t) \le \frac {1} {1 + \frac {1} {R_{H'_s}(s,u) + 1 + R_{H'_t}(v,t)}} \le 0.99$.
   
   More formally, $1 / R_{G'}(s,t) = \underset{\phi \in \mathbb{R}^{V}: \phi(s)=1, \phi(t)=0}{\min} \sum_{(a,b) \in E} (\phi(a)-\phi(b))^2$ from Property 4 of Lemma~\ref{lem:basic_facts}. Since $\phi(s)$ and $\phi(t)$ are fixed, 
   \begin{align*}
       & \underset{\phi \in \mathbb{R}^{V}: \phi(s)=1, \phi(t)=0}{\min} \sum_{(a,b) \in E} (\phi(a)-\phi(b))^2 = 1 + \underset{\phi \in \mathbb{R}^{V}: \phi(s)=1, \phi(t)=0}{\min} \sum_{(a,b) \in E(G') \setminus (s,t)} (\phi(a)-\phi(b))^2.
   \end{align*}

   Because there is only one path $s-u-v-t$ between $s$ and $t$ in $G'$ if we ignore the two edges $(s,t)$ and $(u',v')$, we simplify the 2nd term as follows.
   \begin{align*}
       & \underset{\phi \in \mathbb{R}^{V}: \phi(s)=1, \phi(t)=0}{\min} \sum_{(a,b) \in E(G') \setminus (s,t)} (\phi(a)-\phi(b))^2 \\ 
       = & \underset{\phi(u), \phi(v) \in \mathbb{R}}{\min} \biggl[ (\phi(u) -\phi(v))^2 + \underset{\phi(s)=1}{\min} \sum_{(a,b) \in E(H_s')} (\phi(a)-\phi(b))^2 +  \underset{\phi(t)=0}{\min} \sum_{(a,b) \in E(H_t')} (\phi(a)-\phi(b))^2 \biggr] \\ 
       = & \underset{\phi(u), \phi(v) \in \mathbb{R}}{\min} \biggl[ (\phi(u) - \phi(v))^2 + \frac{(1 - \phi(u))^2} {R_{H'_s}(s,u)}  + \frac{\phi(v)^2}{R_{H'_t}(v,t)} \biggr] = \frac{1} {1 + {R_{H'_s}(s,u) + R_{H'_t}(v,t)}},
   \end{align*}
   where the second equation follows from Property 4 of Lemma~\ref{lem:basic_facts} and the last equation holds when $ \phi^*(u) = \frac{1 + R_{H'_t}(v,t)}{1 + R_{H'_s}(s,u) + R_{H'_t}(v,t)}$ and $\phi^*(v) = \frac{R_{H'_t}(v,t)}{1 + R_{H'_s}(s,u) + R_{H'_t}(v,t)}$.
   
\end{claimproofof}

\section*{Acknowledgement}
We thank Jingcheng Liu (Nanjing University) for suggesting the proof of Lemma 5 using the characteristic polynomial method. Also, we thank all anonymous reviewers for the helpful comments.

\bibliography{cite}

\newpage

\appendix

\section{Lower Bounds for Large Approximation Ratio} \label{sec:large_constant_degree}
In this section, we prove lower bounds for large approximation ratios. First, we restate Theorem~\ref{thm:lower_bound_large_degree}.

\begin{theorem}\label{thm:formal_lower_bound}
    There are $c_0>0$ and infinitely many $n$ such that given any $d \in [4,n]$ and any $\ell \in [4,n]$, for graphs of $n$ vertices and degree $d$, any local algorithm to approximate $R_G(s,t)$ with success probability $0.6$ and approximation ratio $1+c_0 \cdot \min\{ d,\ell \} $ needs $\Omega(dn/\ell)$ queries.
\end{theorem}
We remark that $d$ and $\ell$ could depend on $n$ in the above theorem as long as there are expander graphs whose size is in $[0.1n,0.5n]$ and regular degree is in $[0.1d,0.9d]$. There are two ways to interpret this theorem. First, let us consider regular graphs of degree $d$ and $\ell=d$. Observe that Property 3 of Lemma~\ref{lem:basic_facts} implies a trivial approximation algorithm: outputting 1 always has an approximation ratio at most $\frac{2}{1/d(s)+1/d(t)}=d$. Our lower bound implies that to improve the trivial factor to $1 + c_0 \cdot d$, the algorithm needs $\Omega(n)$ queries. Secondly, it provides a trade-off between the approximation ratio and query complexity. If we consider graphs with $d>\ell$, the lower bound implies that a $(1+ c_0 \cdot \ell)$-approximation algorithm needs $\Omega(m/\ell)$ queries.


The proof of Theorem~\ref{thm:formal_lower_bound} follows the same line as Theorem~\ref{thm:lower_bound_small_degree}. We pick a regular expander $H$ of size $\le n/2$ and degree $\Theta(d)$. Then we make two copies of it containing $s$ and $t$ separately, say $H_s$ and $H_t$ like Figure~\ref{ex:dumbshell}. Next we connect $s$ and $t$ and add $\Omega(dn)$ extra edges in $H_s$ and $H_t$ as the base graph $G$. Then we construct $G'$ by removing $2\ell$ extra edges in $H_s$ and $H_t$ and adding another $2\ell$ edges between $H_s$ and $H_t$ such that it is difficult to distinguish between $G$ and $G'$. The final calculation shows $R_G(s,t)=1$ and $R_{G'}(s,t)=1/(1+\Omega(\min\{d,\ell\})$.

\begin{proof}
We provide a distribution $\mathcal{G}$ of graphs. By Lemma~\ref{Yao_minimax_principle} of Yao's minimax principle, we only need to consider deterministic algorithms with $q_n$ neighbor queries whose its approximation ratio is $1+ c_0 \cdot \min\{ \ell, d \}$ over $\mathcal{G}$ for some constant $c_0$. Our goal is to prove $q_n=\Omega_{c_0}(m/\ell)$.

    Given $n$ and $d$, there exists an expander $H$ of size $n_H \in [0.25n, 0.5n]$ and degree $d_H=[3d/4]$ with the 2nd eigenvalue of its Laplacian at least $\lambda_H=\min\big\{d_H- 2.01 \sqrt{d_H-1},d_H/2\big\}$. For $d<\log^2 n$, we apply \cite{MODP_22_every_degree} to obtain such an expander; otherwise a random $d$-regular graph satisfying this requirement with high probability \cite{Friedman08}. For ease of exposition, we assume $n_H=n/2$; otherwise we add some dummy vertices to make up $n_H$.

    \begin{enumerate}
        \item We take two copies of $H$ as $H_s$ and $H_t$ that contain $s$ and $t$ separately. 
        \item The vertex set of $G$ is defined to be the union of the vertex sets of $H_s$ and $H_t$.
        \item The edge set of $G$ is defined to be  the union of $(s,t)$, the edge sets of $H_s$, $H_t$, and extra edges in $H_s$ and $H_t$ separately such that $G$ is a $d$-regular simple graph. For convenience, we use $E_s$ and $E_t$ to denote the extra edges, whose sizes are $\Omega(dn)$ since $d_H=[3d/4]$.
    \end{enumerate}

    So the effective resistance between $s$ and $t$ in $G$ is $R_G(s,t)=1$.  Next we construct a random graph $G'$ based on $G$.

    \begin{enumerate}
        \item We delete $\ell$ random edges $(u_1,u'_1),\ldots,$ $(u_\ell,u'_\ell)$ from $E_s$ and another $\ell$ random edges $(v_1,v'_1),$ $\ldots,(v_\ell,v'_\ell)$ from $E_t$ separately (without replacement). 
        \item We add $(u_1,v_1), (u'_1,v'_1),$ $\ldots, (u_\ell,v_\ell), (u'_\ell,v'_\ell)$ to $G'$.
    \end{enumerate}

    We show that the effective resistance between $s$ and $t$ in $G'$ is $O(1/\min\{\ell, d\})$. For completeness, we provide a calculation in Appendix~\ref{sec:omitted_proofs2}.

    \begin{claim} \label{clm:er_constant_degree}
        $R_{G'}(s,t)$ is less than $\frac{1}{1+ c_0 \min\{{\ell},{d}\}}$ for some universal constant $c_0$.
    \end{claim}

    Now we fix a deterministic algorithm $A$ with approximation ratio $1 + c_0 \min\{ \ell,d \}$ and consider the underlying distribution $\mathcal{G}$ is $G$ or $G'$ with probability $1/2$ separately. For convenience, let us modify $A$ such that its output is a claim of the underlying graph is $G$ or $G'$. By Lemma~\ref{Yao_minimax_principle} of Yao's minimax principle,
    \begin{equation}\label{eq:dist_prob2}
        0.6 \le \Pr_{\mathcal{G}}[A \text{ succeeds}]=\frac{\Pr[A(G)=G]}{2} + \frac{\Pr[A(G')=G']}{2}.
    \end{equation}
    
    Next we consider all $q_n$ neighbor queries made by $A$ when the underlying graph is $G$. If $G'$ and $G$ provide the same answers on these neighbor queries, $A$ fails to distinguish them. But $G'$ are obtained from $G$ by modifying $2 \ell$ edges in $E_s$ and $E_t$. So we bound
    \begin{align*}
    \Pr[A(G')=G'] & \le \Pr[A(G)=G'] + \Pr[\text{One neighbor query returns different values}] \\
    & \le 1-\Pr[A(G)=G] + q_n \cdot 2\ell / \min\{|E_s|,|E_t|\}.  
    \end{align*}
    Plugging this into \eqref{eq:dist_prob2}, we obtain $q_n = \Omega(\frac{n d}{\ell})=\Omega(m/\ell)$.
\end{proof}

\subsection{Proof of Claim~\ref{clm:er_constant_degree}}\label{sec:omitted_proofs2}
We finish the calculation of Claim~\ref{clm:er_constant_degree} in this section.

\begin{claimproofof}{Claim~\ref{clm:er_constant_degree}} We consider $\min \limits_{\phi:\phi(s) = 1, \phi(t) = 0} {\sum \limits_{(u,v) \in E(G')} (\phi(u) - \phi(v))^2}$ to upper bound $R_{G'}(s,t)$ $= 1 / \min \limits_{\phi} {\sum \limits_{(u,v) \in E(G')} (\phi(u) - \phi(v))^2}$ in this section. For convenience, we only consider $\phi$ with $\phi(s)=1$ and $\phi(t)=0$ and call $\phi$ a potential function in this section.

In fact, the harmonic function $\phi$ gives the minimum value \[ \min \limits_{\phi:\phi(s) = 1, \phi(t) = 0} {\sum \limits_{(u,v) \in E(G')} (\phi(u) - \phi(v))^2}, \] i.e., $\phi$ satisfies $\phi(v)=\frac{1}{d(v)} \sum_{(u,v) \in E(G')} \phi(u)$. Since we do not have an explicit description of $G'$, we provide a lower bound for any $\phi$ such that it gives an upper bound on $R_G(s,t) = 1 / \min \limits_{\phi} {\sum \limits_{(u,v) \in E} (\phi(u) - \phi(v))^2}$. Because the harmonic one is in $[0,1]^{V}$, let us fix a potential function $\phi$ whose $\phi(s)=1$, $\phi(t)=0$, and $\phi(v) \in [0,1]$ for all $v \in V$. There are two cases.


\begin{enumerate}
    \item More than three quarters of vertices in $u_1, u'_1, \cdots, u_{\ell}, u'_{\ell}$ have potential values $\ge 0.55$; and more than three quarters of vertices in $v_1, v'_1, \cdots, v_{\ell}, v'_{\ell}$ have potential values $\le 0.45$. So there are at least $\ell$ edges between them with a potential difference at least $0.1$. Then 
    \[
    \sum \limits_{(u,v) \in E(G')} (\phi(u) - \phi(v))^2 \ge 1 + \ell \cdot (0.55-0.45)^2 = 1 + 0.01 \ell.\]

    \item Otherwise, at least one quarter of vertices in $u_1, u'_1, \cdots, u_{\ell}, u'_{\ell}$ or $v_1, v'_1, \cdots, v_{\ell}, v'_{\ell}$ that don't satisfy the above condition. Without loss of generality, we assume one quarter of vertices in $u_1, u'_1, \cdots, u_{\ell}, u'_{\ell}$ have potential values $\phi(v)<0.55$, say $u_1, u_2,\cdots, u_{\ell/2}$. Since $G'$ contains the expander $H_s$, we only consider the energy contributed by $H_s$: 
    \[
    \sum \limits_{(u,v) \in E(G')} (\phi(u) - \phi(v))^2 \ge 1 + \sum \limits_{(u,v) \in E(H_s)} (\phi(u) - \phi(v))^2 = 1 + \phi\big( V(H_s) \big)^\top \cdot L_{H_s} \cdot \phi\big( V(H_s) \big).
    \]
    To use the property $\lambda_2(L_{H_s}) \ge \max\{d_H - 2.01 \sqrt{d_H - 1}, d_H/2\} = \Omega(d)$ given $d_H=[3d/4]$, we consider $\phi' \in \mathbb{R}^{V(H_s)}$ defined as $\phi':=\phi\big( V(H_s) \big) - c \cdot \mathbf{1} \big( V(H_s)  \big)$ for $c=\frac{\big\langle \phi(V(H_s)), \mathbf{1} ( V(H_s)  ) \big\rangle}{n/2}$. So
    \[
        \phi\big( V(H_s) \big)^\top \cdot L_{H_s} \cdot \phi\big( V(H_s) \big) = (\phi')^{\top} \cdot L_{H_s} \cdot \phi' \ge \lambda_2(L_{H_s}) \cdot \|\phi'\|_2^2.
    \]
    Next we lower bound $\| \phi' \|_2^2$.
    \begin{align*}
        \| \phi' \|_2^2 & = \sum_{i \in V(H_s)} (\phi(v)-c)^2 \\
                        & \ge (\phi(s)-c)^2 + (\phi(u_1)-c)^2 + \cdots + (\phi(u_{\ell/2})-c)^2 \\
                        & \ge \min_{\beta} \left\{ \left(\phi(s) - \beta\right)^2 + \left(\phi(u_1) - \beta \right)^2 + \cdots + \left( \phi(u_{\ell/2}) - \beta \right)^2 \right\} \\
                        & = \left(\phi(s) - \beta^* \right)^2 + \left(\phi(u_1) - \beta^* \right)^2 + \cdots + \left( \phi(u_{\ell/2}) - \beta^* \right)^2 \\
                        & \qquad \qquad \text{ where } \beta^*=\frac{\phi(s) + \phi(u_1)+ \cdots + \phi(u_{\ell/2})}{\ell/2+1}\\
                        & \ge (\phi(s)-\beta^*)^2\\
                        & = (1-\beta^*)^2 \qquad \qquad \bigg(\text{ use } \beta^* \le 0.7 \text{ since } \phi(u_i)<0.55 \text{ and } \ell \ge 4\bigg)\\
                        & \ge (1 - 0.7)^2 = \Omega(1).
    \end{align*}
    From all discussion above, $\sum \limits_{(u,v) \in E(G')} (\phi(u) - \phi(v))^2 = 1 + \lambda_2(L_{H_s}) \cdot \Omega(1) = 1 + \Omega(d)$. 
\end{enumerate}

So we have $\min \limits_{\phi:\phi(s) = 1, \phi(t) = 0} {\sum \limits_{(u,v) \in E(G')} (\phi(u) - \phi(v))^2} = 1 +  \Omega(\min\{ \ell, d\})$ and $R_G(s,t) \le \big(1 + \Omega(\min\{ \ell, d\}) \big)^{-1}$. 
\end{claimproofof}

\section{Total Effective Resistance}\label{sec:total_eff_res}
We show strong lower bounds for approximating the total effective resistance in this section. Recall that the total effective resistance of a graph $G$ is $R_{tot}(G)=\sum_{u<v} R_G(u,v)$ over all pairs $u$ and $v$. We restate Theorem~\ref{thm:inapprox_total_eff_res}.

\begin{theorem}\label{thm:inapprox_total_eff_res}
For every $n$, even for degree-2 graphs, it takes $\Omega(n)$ queries for any algorithm to give a less than $2$-approximation of the total effective resistance with probability 0.6.

For every $n$ and any $\ell=o(n)$, there exists $m$ such that it takes $\Omega(m/\ell)$ queries for any algorithm to give an $\ell$-approximation of the total effective resistance with probability 0.6. 
\end{theorem}

One remark is that this shows that approximating total effective resistance is very different from approximating local effective resistances. For graphs of degree $2$ (except $s$ and $t$), Appendix~\ref{sec:approx_alg} provides a sublinear-time algorithm to approximate the effective resistance of an adjacent pair $(s,t)$. However, the first part of Theorem~\ref{thm:inapprox_total_eff_res} rules out sublinear-time algorithms with approximation ratio $<2$ for the total effective resistance even for graphs of degree at most $2$. 
Furthermore, the second part of Theorem~\ref{thm:inapprox_total_eff_res} rules out trivial approximations like Property 3 of Lemma~\ref{lem:basic_facts} with ratio $\frac{2}{1/d(s)+1/d(t)}$.

Another remark is about approximating $\sum_{i>1} \frac{1}{\lambda_i(L_G)}$. Since $R_{tot}(G) = n \cdot \sum_{i>1} \frac{1}{\lambda_i(L_G)}$ from Property 1 of Lemma~\ref{lem:basic_facts}. Our result shows that there is no sublinear time algorithm for approximating $\sum_{i>1} \frac{1}{\lambda_i(L_G)}$ even for graphs of degree at most 2. This is in contrast to approximating $(\lambda_1,\ldots,\lambda_n)$ where recent works \citep{CKSV18,BKM22} have provided sublinear time approximation algorithms. For example, the algorithm of \cite{CKSV18} outputs a succinct representation of an approximation  $(\tilde{\lambda}_1,\ldots,\tilde{\lambda}_n)$ such that $\sum_{i=1}^n |\tilde{\lambda}_1 - \lambda_i(\tilde{L}_G)| \le \epsilon \cdot n$ in time $2^{O(1/\epsilon)}$ for any $\epsilon$. Finally, we observe that for expander graphs, algorithms in \cite{peng2021local,li2023new} directly imply that one can $(1+\epsilon)$-approximate $R_{tot}(G)$ in $O(\poly(\frac{\log n}{\epsilon}))$ time. The above Theorem~\ref{thm:total_tight_lower_bound} says for general graphs, such an approximation with similar query complexity is not possible.


We finish the proof of Theorem~\ref{thm:inapprox_total_eff_res} in this section. The first part considers the difference between a ring graph of $n$ vertices and a path graph of $n$ vertices. The second part follows the line of Theorem~\ref{thm:lower_bound_large_degree} --- we add $\ell$ random edges to reduce the total effective resistance by a factor of $\ell$. 

\begin{proof}
    By Yao's minimax principle, we only need to consider deterministic algorithms in this proof. For the 1st part, we fix a ring graph $G_r$ of $n$ vertices. Then we construct a random path graph $G_p$ of $n$ vertices that is obtained from the ring $G_r$ by randomly removing one edge.

We will use the following fact
            \[ \sum\limits_{i=1}^{n} i = \frac {n (n+1)} {2}, \quad \sum\limits_{i=1}^{n} i^2 = \frac {n(n+1)(2n+1)} {6}, \quad \sum\limits_{i=1}^{n} i^3 = \frac {n^2(n+1)^2} {4} \]
    to calculate the total effective resistances of $G_p$ and $G_r$ directly. 
    It holds that 
    \begin{align*}
       R_{tot}(G_p) & = \sum\limits_{1 \le i \le n-1} (n-i) \cdot i = n \sum\limits_{i=1}^{n-1} i -  \sum\limits_{i=1}^{n-1} i^2 =\frac {n^3 - n} {6} \\
       R_{tot}(G_r) & = \sum\limits_{1 \le i \le n-1} (n-i) \cdot \frac {1} {\frac {1} {i} + \frac {1} {n - i}}. \\
       & = \frac {1} {n} \cdot \sum\limits_{1 \le i \le n-1} i \cdot (n - i)^2 = n \sum\limits_{i=1}^{n-1} i - 2 \sum\limits_{i=1}^{n-1} i^2 + \frac {1} {n} \sum\limits_{i=1}^{n-1} i^3 = \frac {n^3 - n} {12}.
    \end{align*}

    So $R_{tot}(G_p) = 2 \cdot R_{tot}(G_r)$. Thus, if an algorithm wants to give a $<2$-approximation for $R_{tot}$, it have to distinguish $G_p$ from $G_r$, which needs $\Omega(n)$ queries. 
    More formally, let the underlying graph $\mathcal{G}$ be $G_r$ or $G_p$ with probability $0.5$ separately. Let $A$ be a deterministic algorithm with approximation ratio $<2$ and success probability $0.6$ over $\mathcal{G}$. Then we modify $A$ to distinguish between $G_r$ and $G_p$ since $R_{tot}(G_r)/R_{tot}(G_p)=2$. So we have 
    \[
    0.6 = \Pr[A \text{ succeeds}] = \Pr[A(G_r)=G_r]/2 + \Pr[A(G_p)=G_p]/2.
    \]
    Let $q$ be the number of queries made by $A$ when the underlying graph is $G_r$. $A$ distinguish them only if one degree query returns 1.
    So
    \begin{align*}
        \Pr[A(G_p)=G_p] & \le \Pr[A(G_r)=G_p] + \Pr[\text{one degree query return 1}] \\
        & \le 1 - \Pr[A(G_r)=G_r] + 2 \cdot q/n.
    \end{align*}
    This shows $q=\Omega(n)$.

    The proof of the 2nd part is similar to the proof of Theorem~\ref{thm:lower_bound_large_degree}. We provide a distribution $\mathcal{G}$ of graphs with degree at least $\ell$. By Lemma~\ref{Yao_minimax_principle} of Yao's minimax principle, we only need to consider deterministic algorithms with $q_n$ neighbor queries whose approximation ratio is $O(\ell)$ over $\mathcal{G}$. Our goal is to prove $q_n=\Omega(m/\ell)$.

    Given $d$, there exists an expander $H$ of size $n_H$ and degree $d_H=[3d/4]$ with the 2nd eigenvalue of its Laplacian at least $\lambda_H=\min\big\{d_H- 2.01 \sqrt{d_H-1},d_H/2\big\}$. For $d<\log^2 n$, we apply \cite{MODP_22_every_degree} to obtain such an expander; otherwise a random $d$-regular graph satisfying this requirement with high probability\cite{Friedman08}.

    \begin{enumerate}
        \item We take two copies of $H$ as $H_s$ and $H_t$ that contain $s$ and $t$ separately. 
        \item The vertex set of $G$ is defined to be the union of the vertex sets of $H_s$ and $H_t$. So $n = 2 n_H$.
        \item The edge set of $G$ is defined to be  the union of $(s,t)$, $E(H_s)$, $E(H_t)$, and extra edges in $H_s$ and $H_t$ separately such that $G$ is a $d$-regular simple graph. For convenience, we use $E_s$ and $E_t$ to denote the extra edges, whose sizes are $\Omega(dn)$ since $d_H=[3d/4]$.
    \end{enumerate}

    So any effective resistance between pairs $(u,v)$ in different expanders, $R_G(u, v) \ge R_G(s, t) \ge 1 $. Then $R_{tot}(G) \ge \sum\limits_{u \in H_s, v \in H_t} R_G(u, v) \ge n^2_H = \Omega(n^2)$.
    
    Next we construct a random graph $G'$ based on $G$.
    
    \begin{enumerate}
        \item We delete $\ell$ random edges $(u_1,u'_1),\ldots,$ $(u_\ell,u'_\ell)$ from $E_s$ and another $\ell$ random edges $(v_1,v'_1),$ $\ldots,(v_\ell,v'_\ell)$ from $E_t$ separately (without replacement). 
        \item We add $(u_1,v_1), (u'_1,v'_1),$ $\ldots, (u_\ell,v_\ell), (u'_\ell,v'_\ell)$ to $G'$.
    \end{enumerate}
    
    Similar to Claim~\ref{clm:er_constant_degree}, we have that $R_{G'}(a,b)=O(\ell^{-1})$ for any vertex $a \in H_s$ and $b \in H_t$.

    Also, the effective resistance between any pair $u, v$ among the same part has $R_{G'}(u, v) \le \frac{d_H}{\lambda_H} \cdot \frac{2}{d_H} = \frac{1}{\lambda_H} = O(\frac{1}{\ell})$ by Property 3 of Lemma~\ref{lem:basic_facts}.  Then \[ R_{tot}(G') = \sum\limits_{u \in H_s, v \in H_t} R_G(u, v) + \sum\limits_{u, v \in H_s} R_G(u, v) + \sum\limits_{u, v \in H_t} R_G(u, v) = O( \frac{n^2}{\ell}).\]
    
    Now we fix a deterministic algorithm $A$ with approximation ratio $\ell$ and consider the underlying distribution $\mathcal{G}$ is $G$ or $G'$ with probability $1/2$ separately. For convenience, let us modify $A$ such that its output is a claim of the underlying graph is $G$ or $G'$. By Lemma~\ref{Yao_minimax_principle} of Yao's minimax principle,
    \begin{equation}\label{eq:dist_prob3}
        0.6 \le \Pr_{\mathcal{G}}[A \text{ succeeds}]=\frac{\Pr[A(G)=G]}{2} + \frac{\Pr[A(G')=G']}{2}.
    \end{equation}
    
    Next we consider all $q_n$ neighbor queries made by $A$ when the underlying graph is $G$. If $G'$ and $G$ provide the same answers on these neighbor queries, $A$ fails to distinguish them. But $G'$ are obtained from $G$ by modifying $2 \ell$ edges in $E_s$ and $E_t$. So we bound
    \begin{align*}
    \Pr[A(G')=G'] & \le \Pr[A(G)=G'] + \Pr[\text{One neighbor query returns different values}] \\
    & \le 1-\Pr[A(G)=G] + q_n \cdot 2\ell / \min\{|E_s|,|E_t|\}.  
    \end{align*}
    Plugging this into \eqref{eq:dist_prob3}, we obtain $q_n = \Omega(\frac{n d}{\ell})=\Omega(m/\ell)$.
\end{proof}

\section{Approximation Algorithm for Degree $2$} \label{sec:approx_alg}

In Section~\ref{sec:lower_bound_deg_3}, we proved a lower bound for degree-$3$ graphs (except $s$ and $t$). Now we give an approximation algorithm to solve the case of degree 2 in sublinear time. Since all vertices are of degree $2$ except $s$ and $t$, essentially, $G$ is constituted by several disjoint paths (some of them are disconnected) between $s$ and $t$.

\begin{restatable}[]{theorem}{upperbounddegreetwo} \label{thm:upper_bound_degree_2}
    Given any $\epsilon < 0.1$, for any graph with degree at most 2 except the adjacent pair $s$ and $t$, there is a local algorithm such that with probability 0.99, it outputs a $(1+\epsilon)$-approximation of $R_G(s,t)$ in time $O\bigg( \cfrac {\min\{d(s)^2,d(t)^2\} \cdot \log^2 n \cdot \log \log n} {\epsilon^2} \bigg)$.
\end{restatable}

Note that as long as $\min\{d(s),d(t)\}$ is not too large, the running time in the above theorem is sublinear for any constant $\epsilon>0$. In the rest of this section, we finish the proof of Theorem~\ref{thm:upper_bound_degree_2}. In fact, our algorithm could output an additive-error approximation $\tilde R_{G}(s,t)$ in time $O\bigg( \cfrac {\min\{d(s),d(t)\} \cdot \log^2 n \cdot \log \log n} {\delta \epsilon} \bigg)$ such that $\tilde R_G(s,t)=(1\pm \epsilon) \cdot R_G(s,t) + \delta$. Since $R_G(s,t) \ge \frac{2}{1/d(s)+1/d(t)} \ge 1/\min\{d(s),d(t)\}$ from Property 3 of Lemma~\ref{lem:basic_facts}, $\tilde R_{G}(s,t)$ is a $(1+2\epsilon)$-approximation if we choose $\delta=\epsilon/\min\{d(s),d(t)\}$. This choice of $\delta$ gives the running time in Theorem~\ref{thm:upper_bound_degree_2}. Then we show the additive-error approximation.

First of all, by the following claim (whose proof is deferred to Appendix~\ref{sec:proof_additive_error}), we focus on the approximation of $\rho:=1/R_G(s,t)$ instead of the approximation of $R_G(s,t)$.
\begin{claim}\label{clm:inverse_additive_error}
If $\tilde{\rho} = (1 \pm \epsilon) \cdot \rho \pm \delta$ for $\epsilon<0.1$, $\delta<0.1$, and $\rho \ge 1$, we always have 
 that $1/\tilde{\rho}$ is a $(1+2\epsilon, 2\delta)$-approximation of $1/\rho$. 
\end{claim}
We remark that $\rho \ge 1$ for adjacent pairs. But for any non-adjacent pair $s$ and $t$, our algorithm gives an additive-error approximation on their conductance $1/R_G(s,t)$.

Because all vertices except $s$ and $t$ are of degree $2$, $G$ is constituted by several \emph{disjoint} paths (some are disconnected) between $s$ and $t$. For convenience, we assume $d(s) \le d(t)$ and consider each edge from $s$ as a path. In this proof, we define its length $\ell_i=+\infty$ if it is disconnected; otherwise $\ell_i$ is the exact length of that path from $s$ to $t$. Since $R_G(s,t)=\frac{1}{\sum_i 1/\ell_i}$, our goal becomes to approximate $1/R_G(s,t)=\rho=\sum_{i} 1/\ell_i$. In Algorithm~\ref{alg:exam1}, we describe the additive-error approximation algorithm motivated by the chaining argument whose idea is to sample longer paths with smaller probability. Its guarantee is 
\begin{equation}\label{eq:additive_error}
\tilde \rho = (1\pm \epsilon) \cdot \rho \pm \delta \qquad \text{ for the output } \tilde \rho  = 1/\tilde R_G(s,t) \text{ and } \rho=1/R_G(s,t),    
\end{equation}
which gives the additive error approximation $\tilde R_G(s,t) = (1 \pm 2\epsilon) \cdot R_G(s,t) \pm 2\delta$ by Claim~\ref{clm:inverse_additive_error}.

\begin{algorithm}[H]
	\caption{Algorithm for estimating effective resistance with additive error} \label{alg:exam1}
	\begin{algorithmic}[1]
		\Function{AdditiveApproximationEffectiveResistance}{$G, \epsilon, \delta,s,t$}
            \State If $d(s)>d(t)$ then swap $s$ and $t$
            \State $\tilde \rho \leftarrow 0, p_0 \leftarrow 1$    
            \State $a \leftarrow \cfrac {20 \log n \cdot \log \log n} {\epsilon \delta} $ 
            \For {$k$ \textbf{in} $[0,\cdots,\log n]$}
                \For {each path $i$ from $s$}
                    \State With prob.~$p_k$, take at most $2^{k+1} \cdot a$ steps in path $i$ to find out its length $\ell_{i}$ 
                    \State if it reaches $t$ and $\ell_{i} > 2^k \cdot a$, then $\tilde \rho \leftarrow \tilde \rho + 1 / (p_k \cdot \ell_i)$
                \EndFor
                \State $p_k \leftarrow p_{k-1} / 2$
            \EndFor
            \State \Return $1/\tilde \rho$ as the effective resistance (and $\tilde \rho$ as the conductance)
		\EndFunction
	\end{algorithmic}
\end{algorithm}

In the rest of this section, we prove the correctness of Algorithm~\ref{alg:exam1}. We remark that $p_k=2^{-k}$ in this section. Let us consider its expected time. For a path of length $\ell_i$, if $j=\lfloor \log_2 \ell_i/a \rfloor$, the number of expected steps on this path is at most 
\[
p_0 \cdot 2a + p_1 \cdot 4a + \cdots + p_j \cdot 2^j a + p_{j+1} \cdot \ell_i + \cdots + p_{\log n} \cdot \ell_i = 2a \cdot j + \ell_i \cdot (p_{j+1} + \cdots + p_{\log n}) = O(a \cdot \log n).
\]
Since we can start from either $s$ or $t$, the expected time is $O(\min\{d(s),d(t)\} \cdot a \cdot \log n)$.

Next we show the approximation guarantee in \eqref{eq:additive_error}. For all paths whose lengths are at most $2 a$, the algorithm gets the exact conductance $1/\ell_i$ of them. The rest paths are divided into $\log n$ groups $G_1,\ldots,G_{\log n}$. The paths in group $G_k$ have lengths between $(2^k \cdot a,2^{k+1} \cdot a]$. Assuming path $i$ is in group $k$, we define the random variable $X_i$ for estimating $1/\ell_i$:
\begin{equation*}
    X_i = 
    \left\{
        \begin{array}{ll}
            \cfrac {1} {p_k \ell_i} - \cfrac {1} {\ell_i} \qquad & \textrm{with  probability  $p_k$}, \\
            - \cfrac {1} {\ell_i} \qquad & \textrm{with  probability  $1 - p_k$}.
        \end{array}
    \right.
\end{equation*}

So $\E X_i = 0$ and $\Var X_i = \cfrac {1 - p_k} {p_k \cdot \ell_i^2} \le \cfrac {1} {p_k \cdot \ell_i^2}$. And the result of Algorithm \ref{alg:exam1} is $\tilde \rho = \sum \cfrac {1} {\ell_i} + X_i = \rho + \sum X_i$. 

Now, we calculate the approximation of group $G_k$ by Bernstein's inequality. There are two cases of $G_k$, i.e., $\sum\limits_{G_k} \cfrac 1 {\ell_i} > b$ and $\sum\limits_{G_k} \cfrac 1 {\ell_i} \le b$ for a threshold value $b = \cfrac {\delta} {\epsilon \cdot \log n}$.

\begin{lemma}[Bernstein's inequality \cite{bernstein1912ordre}]
    Let $X_1, \cdots, X_n$ be independent zero-mean random variables. Suppose that $\left| X_{i} \right|\le M$ almost surely, for all  $i$. Then, for all positive $t$, \[ \Pr \left[ \left| \sum\limits_{i=1}^{n} X_i \right| \ge t \right] \le 2 \exp \left[ - \frac{1}{4} \cdot \min \left\{  { \cfrac {t^2} {\sum\limits_{i=1}^{n} \E X_i^2}, \frac {t} { M} }  \right\}  \right]. \]
\end{lemma}

In the previous case, ${(\sum\limits_{G_k} \cfrac 1 {\ell_i})^2} / {\sum\limits_{G_k} \cfrac 1 {\ell_i^2}} \ge \frac{ (\left | G_k \right | \cdot \min_{G_k} {\ell_i} )^2} {\left | G_k \right | \cdot (\max_{G_k} {\ell_i} )^2} \ge {\left | G_k \right |} / 4 \ge 2^{k-2} \cdot a b$. The first inequality comes from $\max\limits_{G_k} {\ell_i} < 2 \min\limits_{G_k} {\ell_i}$. And the last step is implied by $b < \sum\limits_{G_k} \cfrac 1 {\ell_i} < \left | G_k \right | / (2^k \cdot a)$ since each $\ell_i > 2^k \cdot a$ in $G_k$.

Then $\sum\limits_{G_k} \E X_i^2 \le \sum\limits_{G_k} \cfrac {1} {p_k \ell_i^2} \le \frac{4}{a b} (\sum\limits_{G_k} \cfrac 1 {\ell_i})^2$, and $|X_i| \le \frac{2^k-1}{2^k \cdot a} \le \frac{1}{a}$ for all $i \in G_k$. From the above discussion, Bernstein's inequality infers that for $t = \epsilon \sum\limits_{G_k} \cfrac 1 {\ell_i}$,
\begin{align*}
    \Pr \left[ \left | \sum\limits_{G_k} X_i \right | \ge \epsilon \sum\limits_{G_k} \cfrac 1 {\ell_i}  \right] \le 2 \exp \left[ - \frac{1}{4} \min \left\{  { \frac {ab \epsilon^2} {4},  a \epsilon \sum\limits_{G_k} \cfrac 1 {\ell_i} }  \right\}  \right]  \le 2 \exp \left[ - \frac{1}{4} \min \left\{  \frac {ab \epsilon^2} {4},  a b \epsilon  \right\}  \right].
\end{align*}
The second step is by the assumption $\sum\limits_{G_k} \cfrac 1 {\ell_i} > b$.

In the latter case, $\sum\limits_{G_k} \E X_i^2 \le \sum\limits_{G_k} \cfrac 1 {p_k \ell_i^2} \le \sum\limits_{G_k} \cfrac 1 {a \cdot \ell_i} \le \cfrac {b} {a}$ with $\ell_i \ge 2^k \cdot a$ and $\sum\limits_{G_k} \cfrac 1 {\ell_i} \le b$. So
\[ 
    \Pr \left[  \left | \sum\limits_{G_k} X_i \right | \ge \frac{\delta}  {\log n}  \right] 
    \le 2 \exp \left[ - \frac{1}{4} \min \left\{  \frac{a \delta^2} {b \log^2 n}, \frac{a \delta} {\log n}  \right\}  \right].
\]

Take a union bound for all $G_k$, and recall that $b = \frac{\delta}{\epsilon \cdot \log n}$,
\begin{align*}
    & \Pr \left[ \left | \sum\limits_{i \in [d]} X_i \right | \ge \epsilon \cdot \sum\limits_{i \in [d]} \cfrac {1} {\ell_i} + \delta  \right] \\ 
    \le & 2 \log n \cdot \exp \left[ - \frac{1}{4} \min \left\{  { ab \epsilon^2 / 4,  a b \epsilon }  \right\}  \right] + 2 \log n \cdot \exp \left[ - \frac{1}{4} \min \left\{  \frac{a \delta^2}{b \log^2 n} , \frac{a \delta}{\log n}  \right\}  \right] \\
    & \qquad \qquad (\text{plug } b = \frac{\delta}{\epsilon \cdot \log n})
    \\
    \le & 2 \log n \cdot \exp \left[ - \frac{1}{4} \min \left\{  { \frac{a \delta \epsilon}{4 \log n},  \frac{a \delta}{\log n} }  \right\}  \right] + 2 \log n \cdot \exp \left[ - \frac{1}{4} \min \left\{  \frac{a \delta \epsilon}{\log n} , \frac{a \delta}{\log n}  \right\}  \right] \\ 
    \le & 4 \log n \cdot \exp \left[ -  \frac{a \delta \epsilon}{16 \log n}  \right].
\end{align*}
When $a = \cfrac {20 \log n \cdot \log \log n} {\epsilon \delta}$, $\tilde \rho$ is a $(1+\epsilon, \delta)$-approximation of $\rho$ with probability $0.99$.

\subsection{Proof of Claim~\ref{clm:inverse_additive_error}}\label{sec:proof_additive_error}
We finish the proof of Claim~\ref{clm:inverse_additive_error} in this section. Recall that $\tilde \rho = (1\pm \epsilon) \cdot \rho \pm \delta$ for some $\rho \ge 1$.
\begin{align*}
    1/\tilde{\rho} & =\frac{1}{(1\pm\epsilon) \rho \pm \delta}   \\
    & = \frac{1}{\rho} \cdot \frac{1}{1 \pm \epsilon \pm \delta/\rho} \\
    & = \frac{1}{\rho} \cdot (1 \pm 2\epsilon \pm 2 \delta/\rho) \\
    & = (1 \pm 2 \epsilon) \frac{1}{\rho} \pm 2 \delta.
\end{align*}

\section{Supplemental Proofs}\label{append}

\begin{lemma}[Yao's Minimax Principle] \label{Yao_minimax_principle}
  Let $\mathcal{X}$ be a set on inputs to a problem. Let $R_{\epsilon}$ be the minimal complexity among all randomized algorithms that solve the problem with success probability at least $1 - \epsilon$, for all inputs $x \in \mathcal{X}$. Let $D_{\epsilon}^{\mu}$ be the minimal complexity among all deterministic algorithms that solve the problem correctly on a fraction of at least $1 - \epsilon$ of all inputs $x \in \mathcal{X}$, weighed according to a distribution $\mu$ on the inputs. Then \[ R_\epsilon \ge \max\limits_{\mu} D_{\epsilon}^{\mu}.\]
\end{lemma}

Finally, we present the 3rd proof of Lemma~\ref{lemma:laplacian_eigenvalues_perturbation} based on Weyl's inequality of perturbed eigenvalues.

\begin{proofof}{Lemma~\ref{lemma:laplacian_eigenvalues_perturbation}}
    For ease of exposition, we assume $A$ is of rank $n$ in this proof. Otherwise, let $P$ be the matrix of dimension $\rank(A) \times n$ constituted by an orthonormal basis of the column space of $A$. Then we rewrite $A=A_0 \cdot P$ for a matrix $A_0 \in \mathbb{R}^{m \times \rank(A)}$. Since $\rank(A_0)=\rank(A)$ and $P$ does not change the eigenvalues, we could then work on $A_0$ instead of $A$ in this proof. 

    Recall that for any symmetric matrix $X \in \mathbb{R}^{n \times n}$, we always use $\lambda_1(X) \le \cdots \le \lambda_n(X)$ to denote its eigenvalues in the non-decreasing order.  In the rest of this proof, we fix $k$ and $\ell$. So our goal is to prove that $A'$ --- the submatrix of $A$ removing row $\mathbf{a}_\ell$ --- has $\lambda_k \big( (A')^\top \cdot A' \big) \in \bigg[ (1-\tau_\ell) \cdot \lambda_k(A^\top A), \lambda_k(A^\top A) \bigg]$ where  $\tau_\ell=\mathbf{a}_\ell^\top \cdot (A^\top A)^{-1} \cdot \mathbf{a}_\ell$ is the leverage score of $\mathbf{a}_\ell$.  For convenience, let $L:=A^{\top} \cdot A$ and $\Delta := - \mathbf{a}_{\ell} \cdot \mathbf{a}_{\ell}^{\top}$ such that $(A')^{\top} \cdot A' = L + \Delta$. 
  
    $\lambda_k(L + \Delta)=-\lambda_{n-k+1}(- L - \Delta)$ implies 
    \begin{equation}\label{eq:zero_eigenvalue}
        \lambda_{n-k+1}\left( \lambda_k( L + \Delta) \cdot I - L - \Delta \right)=0.    
    \end{equation}
  
    Since $L=A^\top A \in \mathbb{R}^{n \times n}$ and $\rank(L)=\rank(A)=n$, $L$ is non-singular and its inverse $L^{-1}$ is well-defined. Thus $L^{-1/2}$ is non-singular and symmetric. We will use Sylvester’s law of inertia to keep the zero eigenvalue in \eqref{eq:zero_eigenvalue}.    Let $A,B\in\mathbb{R}^{n\times n}$ be two symmetric matrices. We say that $A$ and $B$ are \emph{congruent} if there exists non-singular matrix $S\in \mathbb{R}^{n\times n}$ satisfying $B=S A S^\top$.

    \begin{lemma}[Sylvester's law of inertia~\cite{franklin2012matrix}] 
    Let $A$ and $B$ be two congruent matrices. Then $A$ and $B$ have the same number of positive, negative and zero eigenvalues.
    \end{lemma}
  
    By Sylvester's law of inertia, the number of positive, negative, and zero eigenvalues of $\lambda_k( L + \Delta) \cdot I - L - \Delta$ and $L^{-1/2} \cdot \big( \lambda_k( L + \Delta) \cdot I - L - \Delta \big) \cdot L^{-1/2}$ are the same. This implies
        \begin{equation}\label{eq:one_eigenvalue}
                \lambda_{n-k+1}\left( L^{-1/2} \cdot \big( \lambda_k( L + \Delta) \cdot I - L - \Delta \big) \cdot L^{-1/2} \right)=0 \text{ from } \eqref{eq:zero_eigenvalue}.          
        \end{equation}

    Now we simplify the matrix
    \[
    L^{-1/2} \cdot \big( \lambda_k( L + \Delta) \cdot I - L - \Delta \big) \cdot  L^{-1/2} = \lambda_k(L+\Delta) \cdot L^{-1} - I - L^{-1/2} \cdot \Delta \cdot L^{-1/2}.
    \]
    After removing $I$, \eqref{eq:one_eigenvalue} implies
    \begin{equation} \label{eq:eigenvalue_1}
        \lambda_{n-k+1} \big( \lambda_k(L+\Delta) \cdot L^{-1} - L^{-1/2} \cdot \Delta \cdot L^{-1/2} \big) = 1.
    \end{equation}

    


    To finish the proof, we state the following version of the Weyl's inequality about the interlacing eigenvalues.

    \begin{lemma}[Weyl's inequality \cite{franklin2012matrix}]\label{lem:weyl_lemma}
        Let $A, B, A+B$ be $n \times n$ Hermitian matrices, with their respective eigenvalues $\mu _{i},\,\nu _{i},\,\rho _{i}$ ordered as follows: 
        \[
        A+B: \quad \mu _{1} \le \cdots \le \mu _{n}, \qquad A: \quad \nu _{1} \le \cdots \le \nu _{n}, \quad B: \qquad \rho _{1} \le \cdots \le \rho _{n}. 
        \]
        
        Then the following inequalities hold:
        \[
        \mu_{j} \ge \nu_{i}+\rho_{j-i+1} \text{ for } i \le j,  \qquad \text{ and }
        \mu_{j} \le \nu_{i}+\rho_{j-i+n} \text{ for } i \ge j.
        \]
    \end{lemma}

    To apply Lemma~\ref{lem:weyl_lemma}, we set $A = \lambda_{k}(L + \Delta) \cdot L^{-1}$ and $B = -L^{-1/2} \cdot \Delta \cdot L^{-1/2}$. From \eqref{eq:eigenvalue_1},
    \[
    A+B = \lambda_k(L+\Delta) \cdot L^{-1} - L^{-1/2} \cdot \Delta \cdot L^{-1/2} \text{ has } \mu_{n-k+1} = 1.
    \]
    Then we apply Lemma~\ref{lem:weyl_lemma} to the eigenvalue $\mu_{n-k+1}=1$ with $i=j=n-k+1$ to conclude:
        \begin{equation}\label{eq:eigenvalues_relation}
            1 \ge \nu_i(A) + \rho_1(B) \text{ and } 1 \le \nu_i(A) + \rho_n(B).
        \end{equation}
    Next, we compute those eigenvalues. Observe that $A=\lambda_{k}(L + \Delta) \cdot L^{-1}$ has
    \[
    \nu_i(A)=\lambda_{k}(L+\Delta) \cdot \lambda_{i}(L^{-1})=\lambda_k(L+\Delta) /\lambda_{n-i+1}(L) = \lambda_k(L+\Delta)/\lambda_k(L),
    \]
    where we use $\lambda_{i}(L^{-1})=1/\lambda_{n-i+1}(L)$ in the first step.
    
    We plug the definition $\Delta=-\mathbf{a}_\ell \cdot \mathbf{a}_\ell^{\top}$ into $B= -L^{-1/2} \Delta L^{-1/2}$ to get
    $
    B=-L^{-1/2} \cdot (- \mathbf{a}_\ell \cdot \mathbf{a}_\ell^{\top}) \cdot L^{-1/2} = (L^{-1/2} \cdot \mathbf{a}_\ell) 
    \cdot (L^{-1/2} \cdot \mathbf{a}_\ell)^{\top}
    $.
    So $B$ is of rank $1$ whose eigenvalues $\rho_1=\rho_2=\cdots=\rho_{n-1}=0$ and $\rho_n=\Tr(B)=\Tr\big( (L^{-1/2} \cdot \mathbf{a}_\ell)^{\top} \cdot (L^{-1/2} \cdot \mathbf{a}_\ell) \big)=\mathbf{a}_\ell^{\top} \cdot L^{-1} \cdot \mathbf{a}_\ell=\tau_\ell$.

    Finally, we plug the above calculations into \eqref{eq:eigenvalues_relation}. So
    \begin{align*}
        1 & \ge v_i(A) + \rho_1(B)=\lambda_k(L+\Delta)/\lambda_k(L)+0 \quad \text{ implies } \quad \lambda_k(L+\Delta)\le \lambda_k(L); \\
        \text{ and }     1 & \le \lambda_k(L+\Delta)/\lambda_k(L) + \tau_\ell \quad \text{ implies }  \quad\lambda_k(L+\Delta) \ge (1-\tau_\ell)\lambda_k(L).
    \end{align*}
\end{proofof}

\end{document}